%% file: main.tex
\documentclass[letterpaper,12pt]{article}
\usepackage[T1]{fontenc}
\usepackage{amsmath}

\AddToHook{cmd/appendix/before}{%
    \crefalias{section}{appendix}%
    \crefalias{subsection}{appendix}
}

\newcommand{\vb}{\vec{b}}
\newcommand{\vi}{\vec{i}}
\newcommand{\vt}{\vec{t}}

\input{preamble.tex}

\title{List Decoding, Linear Hashing, and Furstenberg over $\F_q$}
\author{
Vinayak M. Kumar\thanks{Simons Institute for the Theory of Computing, UC Berkeley. \href{mailto:vmkumar@berkeley.edu}{vmkumar@berkeley.edu}. Supported by a Simons Institute Research Fellowship. Part of this work was conducted while the author was a graduate student at UT Austin, supported by a Jane Street Graduate Research Fellowship and a UT Austin Dean's Prestigious Fellowship Supplement.}
\and Geoffrey Mon\thanks{Department of Computer Science, University of Texas at Austin. \href{mailto:gmon@cs.utexas.edu}{gmon@cs.utexas.edu}. Supported by NSF Award CCF-2505865, an NSF Graduate Research Fellowship (DGE-2137420), and a UT Austin Dean’s Prestigious Fellowship Supplement. Part of this work was conducted while the author was visiting the Simons Institute for the Theory of Computing.}
}
\date{}

\begin{document}

\maketitle

\input{proof-pieces}

\end{document}

%% file: proof-pieces.tex
\begin{abstract}
    We give new bounds for
    list sizes of random linear codes at capacity,
    max loads of linear hash functions,
    and Furstenberg sets, over every finite field $\mathbb{F}_q$.
    \begin{enumerate}
        \item \textbf{Random linear codes} over $\mathbb{F}_q$ with rate $1 - H_q(p) - \eps$
        are $(p, O(q H_q(p)/\eps))$-list decodable with high probability
        for all values of $p, q, \eps$,
        including the high error regime.
        This nearly matches
        the list size lower bound of $H_q(p)/\eps$ due to Guruswami, Li, Mosheiff, Resch, Silas, and Wootters [\emph{IEEE Trans.\ Inf.\ Theory} 2022].
        Our bound is the first uniform improvement for $q > 2$
        since Guruswami, H\aa{}stad, and Kopparty [STOC~2010].

        \item \textbf{Linear hash functions} over $\mathbb{F}_q$
        hashing $n$ balls to $n$ bins achieve maximum load 
        $O(q \ln \ln q / {\ln q}) \cdot \ln n / {\ln \ln n}$, 
        both in expectation and with probability $1-o(1)$. 
       This nearly matches the lower bound of $\ln n / {\ln \ln n}$.
        Previously, only a 
        polylogarithmic upper bound was known for $q > 2$,
        due to Alon, Dietzfelbinger, Miltersen, Petrank, and Tardos [\emph{J. ACM} 1999].
    \end{enumerate}
    We reduce list decodability and linear hashing to strong Furstenberg set lower bounds, which we prove using a new polynomial method of multiplicity gaps.
    While previous polynomial methods analyze a set $S$ by studying polynomials that vanish on it, we consider polynomials that vanish everywhere, but with higher multiplicity inside $S$ than outside.
\end{abstract}

\section{Introduction}
\label{sec:intro}

Subspaces are a fundamental primitive in mathematics and theoretical computer science.
In this paper, we make progress on three intertwined questions on subspaces: list decodability of random linear codes, expected max load of linear hashing, and lower bounds on Furstenberg sets.

\subsection{List Decodability of Random Linear Codes}

An error-correcting code (or simply code) over a finite field $\F_q$ 
is a subset $C\subseteq \F_q^n$ of well-separated strings (known as codewords).
The separatedness of $C$ is measured by its \emph{distance} $d$, which is the minimum Hamming distance between two distinct elements of $C$. In an adversarial (Hamming) error model, where a transmitted codeword could be corrupted on any $< d/2$ coordinates, unique decoding is {always} information-theoretically possible
because the closest codeword is unique. 
However, it is easy to see that this promise breaks when we allow at least $d/2$ corruptions to take place.

A revolutionary idea in coding theory is the notion of \emph{list decoding}:
instead of hoping for a unique closest codeword to a noisy input,
we are satisfied if
the number of codewords close to a noisy string is small.
This weaker guarantee has nonetheless found many applications in both theory and practice,
and has spawned fruitful lines of research.
One such program is to find families
of list decodable
codes that achieve the best possible
parameters.

On the one hand, we want to maximize the \emph{rate}, $R\coloneqq \log_q |C|/n$,
which measures how much information each codeword contains,
because we want to minimize the code's redundancy (i.e., inefficiency).
On the other hand, we want to minimize the \emph{list size}, which is the largest 
number of codewords that are close (say, within distance $pn$) to
a single string in $\F_q^n$.
These two goals are in conflict with each other, because we can always maximize rate by blowing up list size and vice versa.
When the list size is at most $L$,
we say $C$ is \emph{$(p,L)$-list decodable}.

\emph{List decoding capacity}
refers to a particular rate threshold $1 - H_q(p)$
below which (for sufficiently large $n$)
list decoding
with constant list size is possible,
and above which constant list size is impossible. 
A code \emph{achieves capacity}
if it has rate close to this threshold
and constant list size,
and we are interested in which families of codes
do so.
Through a successful series of works,
we now know that one such family
is random linear codes.

\paragraph{Random linear codes.}
Random linear codes (RLCs) are one of the most studied classes of practical codes. Linear codes have efficient encoders and representations, as well as useful properties such as symmetry. These properties make random linear codes helpful,
for example as a building block
for constructing good codes
via concatenation.
Random linear codes are also intimately related to notions in cryptography such as LWE and LPN. 
Hence, it is of great interest to understand the properties of random linear codes. Indeed, there has been a long line of work starting in the 1950s studying 
properties of RLCs, and how close they are in strength to fully random codes.

Let us consider the context of list decoding capacity
in detail.
In general, if a code has
rate near capacity,
i.e.,
$1 - H_q(p) - \eps$,
then its list size
will grow as the \emph{capacity gap} $\eps$ shrinks.
Elias~\cite{elias91}
showed that
a fully random code
with rate gap $\eps$ from capacity
will
be $(p,H_q(p)/\eps)$-list decodable 
with high probability, for all
choices of parameters $q,p,\eps$.
Guruswami, Li, Mosheiff, Resch, Silas, and Wootters~\cite{glmrsw}
also show that
an RLC cannot outperform a fully random code,
because its list size
must be at least $H_q(p)/\eps$.
Thus,
researchers have been pushing towards
showing RLCs
have list decodability
which is 
just as strong as fully random codes.

\paragraph{List size bounds.}
In the binary case ($q=2$),
Guruswami, H\aa{}stad, Sudan, and Zuckerman~\cite{ghsz}
showed the existence
(improved by \cite{liwoo} to high probability) of linear codes with capacity gap $\eps$
which are $(p, O(1/\eps))$-list decodable
using potential functions.
This method was developed further by Li and Wootters~\cite{liwoo} to show random linear codes with rate gap $\eps$ are $(p, H_2(p)/\eps + 2)$-list decodable with high probability, thereby matching fully random codes
up to a small additive constant.
This essentially completely resolves
RLC list decoding for $q=2$.

For $q > 2$, 
we know much less,
as summarized in \cref{tab:list-size}.
A bound of Zyablov and Pinsker~\cite{zp81} gives a bound of $q^{O(1/\eps)}$ for all $p$. Guruswami, H\aa stad, and Kopparty~\cite{ghk} showed that random $\F_q$-linear codes with capacity gap $\eps$ are $(p, C_{p,q}/\eps)$-list decodable with high probability, where 
\[C_{p,q} \approx \exp\parens*{O\parens*{\frac{\log_2^2 q}{\min((1-1/q - p)^2, p)}}},\]
as analyzed by \cite[Theorem 7.1]{GM22}. For fixed $p$ and $q$, this 
bound has the optimal $\eps$ dependence in contrast with Zyablov and Pinsker~\cite{zp81}.
However, this bound's dependence on $q$ is quasipolynomial (worse than \cite{zp81}), and the bound exponentially grows as $p$ approaches $1-1/q$ (the \emph{high error regime}). In particular, if $\delta \coloneq 1-1/q-p$,
note that $C_{p,q}$ is exponential in $1/\delta^2$.
The difference 
between known list sizes for 
$q=2$ versus $q=3$ is quite unsettling---intuitively, incrementing the alphabet size should not drastically affect the list size in such a conceptual way!

\begin{table}[]
    \centering
    \caption{Uniform list decodability bounds for $q \ge 3$ and capacity gap $\eps$
    ($\delta \coloneq 1 - 1/q - p$)}
    \begin{tabular}{cll}
    \toprule
     \textbf{Source} 
     &\textbf{List size bound}\\
    \midrule
    \multicolumn{2}{l}{\emph{Fully random codes}}\\
    \cite{elias91}
    & $\sim H_q(p)/\eps$\\
    \midrule
    \multicolumn{2}{l}{\emph{Random linear codes}}\\
    \cite{glmrsw}
    & $\ge H_q(p)/\eps$
    \\
    \cite{zp81}
    & $\le q^{O(1/\eps)}$\\
     \cite{ghk} 
     & $\le \exp(O(\frac{\log_2^2 q}{\min\{p,\delta^2\}}))/\eps$\\
     \textbf{this work}
     & $\le O(q H_q(p)/\eps)$\\
    \bottomrule
    \end{tabular}
    \label{tab:list-size}
\end{table}

\paragraph{High error regime.} This motivated a line of work
proving list size bounds for $\F_q$-RLCs
with $q \ge 3$
in the high error regime \cite{cgv,Woo13,rw14,RW17}.
Recall that a fully random code has a uniform list size bound of $H_q(p)/\eps$,
\emph{even} if we set parameters
$\delta=o(1),\eps = o(1)$. Can we prove the same bound for random $\F_q$-linear codes in the same regime? Some improvements to \cite{ghk} were made with the condition that $\delta$ and $\eps$ decreased to $0$ in a coupled manner \cite{Woo13,RW17}. But the dream of a uniform list bound of $C/\eps$ for any $p,q,\eps$, and absolute constant $C$ remains open.

\subsection{Linear Hashing}

A more general problem is the question of hashing using a random linear map, which has its own history. 
Hashing is the idea of
randomly mapping a massive universe $U$ to 
a relatively tiny set of $n \ll \abs{U}$ bins.
This compression is necessarily very lossy, and many collisions will occur: many items in the universe will map to the same bin simply by the pigeonhole principle.
But if there is a relatively small set $S$ of items that we actually care about, then hopefully by exploiting the randomness,
collisions between items in $S$
will be rare.
Then, we can use the hash to
pretend that
our universe is essentially of size $n$, and reap the benefits in our algorithm, data structure, etc.

The ubiquitousness of hashing in theoretical computer science cannot be understated. Hash functions appear in data structures or in subroutines of algorithms, as ingredients in pseudorandom constructions, and even as convenient mathematical tools for analyzing objects that have nothing to do with hashing.
A common complaint in these applications is that a completely random function would have amazing properties in these contexts, but at an enormous cost.
Hashing research
seeks efficient hash functions which imitate the properties of these expensive random maps.

One such property of great interest is the maximum number of balls in any bin, known as the \emph{max load},
which is a canonical measure of hashing quality that is critical to randomized algorithm design~\cite[Chapter 5]{mitzupfal}.
Suppose that we have $n$ balls (the relevant items from the universe) and a hash table with $n$ bins. 
Max load measures the imbalance of the table,
which is relevant for,
e.g.,
the performance of hash tables
that use chaining.
In this situation, each bin has an associated linked list that stores all balls that hash to it.
To check if a hash table contains a particular key,
we need to compute its would-be bin
and then sweep through that linked list.
Consequently, the worst-case running time of this operation is proportional to the max load,
so we want it to be small. 

A fully random function 
hashing $n$ balls into $n$ bins will have expected max load $\sim {{\ln n}/{\ln\ln n
}}$.
However,
the space and time
required to deal with such functions is enormous---$O(U \log n)$ bits
are required to even describe such a function, which is large enough to defeat the purpose of hashing.
Can we get similar performance
from a simpler and cheaper family of functions?

\paragraph{Linear hash functions.}
A natural candidate is the random linear map.
For a prime power $q$,
let $n = q^\ell$ for an integer $\ell \ge 1$,
let $u\ge \ell$ be arbitrarily large, and let $h: \F_q^u\to\F_q^\ell$ be a random linear map. For \emph{any} set of $n$ balls $S\subset \F_q^u$, what is the average max load when $h$ hashes $S$ into the set of bins $\F_q^\ell$? This question had been studied since the seminal work of Carter and Wegman introduced universal hash functions \cite{cw79}. 
Linear hash functions
are inexpensive to evaluate and
require only 
$O(\log U \log n)$ bits to describe, so they are an attractive alternative to full randomness.

On the one hand,
we know that linear hash functions cannot have better expected max load
than fully random hash functions.
If we pick the set of balls
completely randomly,
this simulates
the behavior of a 
completely random hash function
and shows that
the expected max load for linear hash functions over any $q$
is at least
$(1-o(1)){\ln n}/{\ln\ln n}$.
In addition,
Alon, Dietzfelbinger, Miltersen, Petrank, and Tardos~\cite{admpt} constructed a set of $n$ balls for which the expected max load is $\Omega(q^{1/3})$, 
which is relevant
when $q$ grows with $n$.

On the other hand,
how much worse
are linear hash functions
compared to fully random hash functions?
Regarding expected max load upper bounds, 
Carter and Wegman~\cite{cw79} established an expected max-load of $O(\sqrt{n})$ on all pairwise independent hash functions. Subsequent
improvements, 
summarized in \cref{tab:hashing},
culminated in Alon, Dietzfelbinger, Miltersen, Petrank, and Tardos~\cite{admpt}, 
whose bound was
later optimized by Babka~\cite{babka18}
to $O(q^{O(\log_2 \log_2 q)}(\log_q n)^{\log_2 q})$.\footnote{\label{footnote:admpt}For $q > 3$, \cite{admpt} only claimed a bound of $O_q((\log n\log\log n)^{\log_2 q})$ without proof, and \cite{babka18} removed the $\log \log n$ term for $q=2$ but did not address the $q>3$ case. We obtain a bound with explicit $q$ dependence by reconstructing the argument we predict they had in mind in \cref{sec:admpt}, with optimizations to remove the $\log\log n$.}
This was the state of affairs until
Jaber, Kumar, and Zuckerman~\cite{jkz25} resolved the $q=2$ case with a bound of $O(\ln n/{\ln\ln n})$;
the leading factor was later eliminated by Bshouty~\cite{bshouty26}.
Peculiarly, however, their proof does not appear to extend to
$q > 2$, and gives \emph{no} bound for these cases---not even for $q=3$!
Hence, the state of the art for $q\ge 3$ remained the polylogarithmic bound of \cite{admpt}. 

\begin{table}[]
    \centering
    \caption{Expected max load bounds for random linear hash functions ($\gtrsim$ and $\lesssim$ hide constants)}
    \begin{tabular}{cll}
    \toprule
     \textbf{Source} & \textbf{Bound for $\bm{q > 2}$} & \textbf{Bound for $\bm{q=2}$}\\
    \midrule
    \multicolumn{3}{l}{\emph{Lower bounds}}\\
    folklore
    & $\ge (1-o(1))\frac{\ln n}{\ln\ln n}$
    & $\ge (1-o(1))\frac{\ln n}{\ln\ln n}$\\
    \cite{admpt}
    & $\gtrsim q^{1/3}$
    & $\gtrsim 1$\\
    \midrule
    \multicolumn{3}{l}{\emph{Upper bounds}}\\
     \cite{cw79} 
     & $\lesssim 
     \sqrt{n}
     $
     & $\lesssim 
     \sqrt{n}
     $ \\
     \cite{mcw78} 
     & $\lesssim 
     q^{1/2} n^{1/4}
     $
     & $\lesssim 
     n^{1/4}
     $ \\     
     \cite{mv81} 
     & $\lesssim q^{O(\sqrt{\log_q n})}$
     & $\le 2^{O(\sqrt{\ln n})}$ \\
     \cite{admpt,babka18}\textsuperscript{\ref{footnote:admpt}}
     & $\lesssim 
     {
     q^{O(\log_2 \log_2 q)}(\log_q n)^{\log_2 q}
     }$%
     & $\lesssim \ln n$ \\
     \cite{jkz25,bshouty26}
     &
     ---
     & $\le (1+o(1)){\frac{\ln n}{\ln \ln n}}$\\
     \textbf{this work}
     & {$\lesssim {
     {\frac{q \ln \ln q}{\ln q} \cdot {\frac{\ln n}{\ln \ln n}}}}$}
     & {$\lesssim {
     {\frac{\ln n}{\ln \ln n}}}$} \\
    \bottomrule
    \end{tabular}
    \label{tab:hashing}
\end{table}

As with list decodability, our techniques thus far are too rigid to work over $\F_3$ or other fields.
In fact, these techniques
require
additional properties,
beyond linearity,
which hold over $\F_2$
but vanish over larger fields.
Indeed, Jaber, Kumar, and Zuckerman~\cite{jkz25} 
comment that their argument works for any combinatorial hash function which iteratively groups universe elements in a $q$-wise independent manner.
This property happens to hold for linear maps over $q=2$, where pairwise independence coincides with the binary field,
but this ``happy coincidence'' ends  at $q\ge 3$.
Thus, we need a new a technique that uses \emph{only} the linearity of a hash function to deduce pseudorandom properties,
rather than extra combinatorial properties which happen
to be true over $\F_2$.

\subsection{Furstenberg Sets}
Finally, consider the fundamental question 
underlying the max load 
of linear hash functions
\emph{and}
the list decodability of RLCs:
determining the optimal size
of a \emph{Furstenberg set} over a finite field.
Before giving the definition,
let us begin with a much more classical notion: Kakeya sets over $\R^n$.
These are subsets that contains a unit line segment in every direction.
The famous Kakeya conecture asks whether Kakeya sets over $\R^n$ must have Hausdorff dimension $n$. This was proven for $n=2$ by Davies \cite{Davies}, and for $n=3$ only very recently by an amazing work of Wang and Zahl \cite{wangzahl}. The proof for general $n$ is a longstanding open question.

A sister question to the Kakeya conjecture is the Furstenberg conjecture asked by Wolff \cite{Wol99}.
This asks what the minimum Hausdorff dimension of a subset $S \subset \R^2$ must be such that for every direction, there exists a line whose intersection with $S$ is lower bounded in Hausdorff dimension. This question was recently resolved by Orponen and Shmerkin \cite{orponen-shmerkin} and by Ren and Wang~\cite{renwang}.

Time has made it apparent that this question becomes very hard for large $n$ (even $n= 3, 4, 5, \dots$). In order to stimulate progress for the high-dimensional case, Wolff \cite{Wol99} introduced the finite field analog of the Kakeya conjecture. If one takes the heuristic that the field $\R^n$ can be approximated by $\F_q^n$ with $q\to \infty$, and sets in $\R^n$ with Hausdorff dimension $d$ correspond to sets in $\F_q^n$ of size $\Theta(q^d)$, the finite field version of the Kakeya conjecture asks whether sets $S\subset \F_q^n $ which contain a line of every slope has size $|S|\ge \Omega_n(q^n)$ (since $q\to\infty$, the hidden constant should not depend on $q$). This was resolved by Dvir \cite{Dvir2009}, whose work kicked off the use of polynomial methods for Kakeya and Furstenberg set lower bounds over finite fields. 

Upon the resolution of Wolff's finite-field Kakeya lower bound conjecture by Dvir \cite{Dvir2009},
generalizations of Kakeya and Furstenberg to higher dimensions were studied, where sets must contain
shifts of every $r$-dimension linear subspace (``slope'') for a fixed parameter $r$.  This notion has seen great attention over the reals, \cite{obe,FalconerMattila2016,HeraKeletiMathe2019,Hera19,dov,bright-dhar}. 
Formally:
\begin{definition}[Kakeya sets]
    A set $S\subset \F_q^n$ is
    \emph{$(\delta, n,r)$-Kakeya}
     if for a $\delta$-fraction of dimension-$r$ subspaces $V$, there exists a shift $x\in \F_q^n$ such that $(x+V) \subseteq S$. If $\delta$ is surpressed, assume $\delta=1$.
\end{definition}
The traditional Kakeya sets over $\F_q^n$ correspond to $\delta=1$ and $r=1$. The study of $(n,r)$-Kakeya
with $r > 1$ was initiated
by Ellenberg, Oberlin, and Tao~\cite{eot}. Let us also formalize Furstenberg:

\begin{definition}[Furstenberg sets]
    A set $S\subset \F_q^n$ is \emph{$(\delta, n, r, T)$-Furstenberg} if for a $\delta$-fraction of dimension-$r$ subspaces $V$, there exists a shift $x\in \F_q^n$ such that $|(x+V) \cap S|\ge T$. If $\delta$ is surpressed, assume $\delta=1$.
\end{definition}
Note that $(n,r)$-Kakeya
corresponds to setting $\delta=1$
and $T = q^r$.
$(n,r,T)$-Furstenberg sets were first studied by Ellenberg and Erman \cite{EE16}.
$(n,r)$-Kakeya and $(n,r,T)$-Furstenberg sets  have received continued attention throughout the past decade  \cite{eot,KLSS, EE16,SS08, BuckChao,ddldrizzy, dhardvir, ddl-simple}, and have 
been used in extractor constructions \cite{dkss, dw11} and lattice coverings \cite{orw22,orw25}.
    Extensions of the polynomial method of Dvir \cite{Dvir2009} led to optimal lower bounds for $(n,r)$-Kakeya sets in various regimes \cite{dkss, KLSS}. The polynomial method, along with point-line incidence theory, gives an array of $(n,r, T)$-Furstenberg set lower bounds \cite{eot,ddldrizzy, dhardvir}. One of our focuses is to prove stronger lower bounds on $(n,r, T)$-Furstenberg sets.
\subsection{Our Results}

We prove strong lower bounds for Furstenberg sets in a broad range of parameters. This is the main technical input for our other results on linear hashing and RLC list decodability.

\begin{theorem}[simplified \cref{thm:furst-dim-reduce}]
\label{thm:intro-main}
Let $\delta > 0$,
$n\ge r\ge 1$, $q$ be a prime power, and $1\le T\le q^r$.
Suppose $S$ is $(\delta, n, n-k, T)$-Furstenberg.
Then,
\[
\abs{S} \ge \frac{\delta T}{8}\cdot 
q^k \exp\parens*{-\frac{8q(k+3\log_q T)}{T}}.
\]
\end{theorem}

We did not attempt to optimize constants in the interest of clarity.
However, the 
leading constant can be optimized at the cost of blowing up the constant inside the exponent, and vice versa.
For an intuitive way of interpreting this bound, first set $\delta = 1$.  Any set of size $Tq^k$ must be $(1,n, n-k, T)$-Furstenberg by the pigeonhole principle. Our lower bound implies that this is optimal up to a mild factor. Furthermore, this bound is independent of $n$, and shines when $k$ and $T$ are small. In our applications, we will have $n$ arbitrarily large, and $k,T$ small, so this is precisely the bound shape that we need.

Let us compare our bound with two relevant known bounds: one in the regime $T$ is large, and the other when $T$ is small and $\delta = 1$. The lower bound of Dhar and Dvir \cite[Theorem 6.14]{dhardvir} is quite weak when $T$ is bounded away from $1$, but has the strong shape of $|S|/q^n\ge \delta Tq^k\exp(-qn/T)$ when $T = 1-o(1)$.  This nearly matches our bound, with the exception of the $n$ in their exponent rather than our (potentially much) smaller $k+3\log_q T$. The other lower bound is of Dhar, Dvir, and Lund \cite[Theorem 2]{ddl-simple}, who show a bound only when $\delta = 1$, in which the shape of $|S|\ge Tq^k(1-\sqrt{q2^{k+7}/T})$ is obtained.
This gives near-pigeonhole bounds for many values of $T$, and achieves a bound independent of $n$ like \cref{thm:intro-main}.
However, this bound is trivial unless $T \gg 2^k$,
and using the heuristic
\[
1-\sqrt{q 2^{k+7}/T} \approx
\exp(-\sqrt{q 2^{k+7}/T}),
\]
this bound has a exponentially
worse dependence on $k$
inside the exponential term.
Our bound unifies and optimizes these two prior bounds. 

The Furstenberg parameter
regime relevant to our applications is the non-classical regime of fixed $q$ and growing $n$.
In this regime, the ``mild'' exponential factor in our bound is actually necessary, up to constants.
Intriguingly, for these parameters,
\emph{random sets} are optimal Furstenberg sets (see \cref{apx:tight}).  This is in stark contrast to Furstenberg and Kakeya set constructions in the classical regime of fixed $n$ and growing $q$, where constructions are algebraic in nature.

\paragraph{Linear hashing.} For our first application, we achieve optimal expected max load when hashing $n$ balls into $n$ bins using a random $\F_q$-linear map.
 For a function $h: A\to B$ and subset $S\subset A$, define \[M(S,h)\coloneqq \max_{b\in B} |h^{-1}(b)\cap S|\] to be the \emph{max load} when $h$ is viewed as a hash function mapping $S\subset A$ into $B$.
 Then,
\begin{theorem}[simplified \cref{thm:eml-final}]
\label{thm:intro-eml}
For sufficiently large $\ell$, the following holds. For \emph{any} prime power $q$, integers $u\ge \ell$ and $n\coloneqq q^\ell$,  subset $S\subset \F_q^u$ of size $n$, and uniformly random linear map $h:\F_q^u\to\F_q^\ell$,
we have \[\E_h [M(S,h)]\le \frac{11q\ln\ln q}{\ln q} \cdot {\frac{\ln n}{\ln\ln n}}.\]
\end{theorem}

Therefore, for constant $q$, we achieve optimal max load (up to constants), resolving a question of \cite{admpt} in this regime. In fact, our techniques also show that optimal max load is achieved with probability $1-o(1)$. For $q=2$, our theorem gives an alternative proof of optimal expected max load to the potential-based argument of Jaber, Kumar, and Zuckerman \cite{jkz25}. 

We note that Alon, Dietzfelbinger, Miltersen, Petrank and Tardos \cite{admpt} show the existence of a subset $S$ of size $q^\ell$ which has max load $\ge q^{1/2}$ when $q$ is a square, and $\ge \Omega(q^{1/3})$ otherwise. Hence, a polynomial dependence on $q$ is unavoidable. However, it is possible the expected max load bound could be improved to $O(q^{1/2} + {\ln n}/{\ln\ln n})$, which would be optimal up to constants for any $q$ and $\ell$. We leave this as an open question. 

\paragraph{List decoding.}
Finally, we show that with high probability, a random $\F_q$-linear code near capacity is list decodable with nearly optimal list size, 
even in the near-maximal error regime.

\begin{theorem}[simplified \cref{thm:list-decoding-final}]
\label{thm:intro-list-decoding}
    For every prime power $q$, $p\in (0,1-1/q)$, $\eps\in (0,1-H_q(p))$, and sufficiently large $n \ge n_0(\eps)$,
    a random linear code $C\le \F_q^n$ of rate $1-H_q(p)-\eps$ is $(p, {O(qH_q(p)}/\eps))$-list decodable with probability $1-q^{-\Omega(\eps n)}$.
\end{theorem}

For all $p,q,\eps$, this is
within a factor $q$ of the lower bound of $H_q(p)/\eps$ \cite{glmrsw}, and thus within a $q$-factor of a dream uniform bound of $O(1/\eps)$ 
that holds for all parameters.
    We emphasize the implications of this bound in the high error rate regime, that is, when $\delta \coloneqq 1-1/q - p \ll 1$. For fixed $q$,
we obtain list size $O(1/\eps)$ for all parameters, even when $\delta\to 0$. To the best of our knowledge, this is the first bound of this type. Previous bounds either blow up (usually exponentially)~\cite{ghk}, or coupled $\eps$ and $\delta$ together~\cite{Woo13,RW17}.
We leave removing this factor of $q$ as an open problem.
 
\subsection{Concurrent Work}
After our work was completed, we learned of two concurrent and independent works regarding the list decodability of random linear codes which were recently posted on the arXiv \cite{yuanzhu,silas}. Yuan--Zhu and Silas both proved that for any $p,q,\eps$, random linear codes of rate $1-H_q(p) - \eps$ are $(p, \frac{H_q(p)}\eps + C_{p,q})$-list decodable for some $C_{p,q}$. Silas further showed that for any fixed $p,q$, for sufficiently small $\eps < \eps_0(p,q)$, a random linear code of rate $1-H_q(p)-\eps$ will have list size either $\floor{H_q(p)/\eps}$
or $\floor{H_q(p)/\eps} + 1$.

 In the work of Yuan and Zhu, Remark 3.9 of their paper declares that their additive term is \[C_{p,q} = \exp\left(O\left(\frac{H_q(p)\ln^2 q}{(1-\frac{1}q-p)^2}\right)\right).\] Hence, this bound has a quasipolynomial dependence on $q$, and blows up exponentially in the high error regime $p\to 1-1/q$. 
 In the work of Silas, Section 8 of the paper claims that for constant $q$, their additive term is \[C_{p,q} = \exp\left(O_q\left(\frac{1}{(1-\frac{1}q-p)^2}\right)\right),\] and so the list size here does blow up exponentially as $p\to 1-1/q$ at a similar rate as Yuan and Zhu. From tracing the $q$ dependence of Silas's $C_{p,q}$, it appears to be the case $C_{p,q} = \exp(O_p(\log^3 q))$, 
 which is a bit worse than that of Yuan and Zhu and of Guruswami, H\aa stad, and Kopparty \cite{ghk}. %

 In our work, we show that for any $p,q,\eps$, an RLC of rate $1-H_q(p) - \eps$ is $(p, 16q({H_q(p)}/\eps + 1))$-list decodable. Our argument can be adjusted to replace the leading constant of $16$ with $1.01$. Hence, for fixed $p,q$ and $\eps\to 0$, our bound is within a factor of $q$ of optimal, making it worse than the bounds of Yuan and Zhu and of Silas in this regime. However, our bound's dependence on $q$ is only a single factor of $q$, and perhaps more notably, our bound has optimal dependence on $p$. No matter how close $p$ gets to $1-1/q$, our list size bound is always at most, say, $16q/\eps$. To the best of our knowledge, our list bound is the 
 only one in the literature which gives a uniform bound of $O_q(1/\eps)$ in the high error rate regime, even with the concurrent work included. 
\subsection{Proof Overview}
\label{sec:proof-overview}

We will first dive into the question of linear hashing and reduce it to Furstenberg set lower bounds. Next, we outline how our strategy for Furstenberg sets, and finally we relate linear hashing to the list decodability of RLCs.

\subsubsection{Linear Hashing}
For simplicity, we will assume our universe has dimension $2\ell$,
and consider a simpler question
for intuition.
For every set $S$ of $n \coloneq q^\ell$
balls from the $n^2$-size universe $\F_q^{2\ell}$,
does there exist a linear map $h\colon \F_q^{2\ell}\to \F_q^\ell$ such that $M(S,h)\le T$, where $T = \Theta_q(\ell/{\ln \ell})$?
Since each $h$ can be identified with its kernel, and each preimage of $h$ is simply a coset of its kernel,
it is sufficient to find one subspace $V\le_\ell \F_q^{2\ell}$ such that for all $x\in \F_q^{2\ell}$, $|(x+V)\cap S|\le T$.

Suppose for contradiction
that there is a subset $S\subset \F_q^{2\ell}$ for which $M(S,h) > T$ for \emph{all} $h$.
This implies that for all $V\le_\ell \F_q^{2\ell}$, there exists a shift $x$ such that $|(x+V)\cap S| > T$. Note that this is \emph{exactly} the definition of a $(2\ell, \ell, T)$-Furstenberg set.
Hence, if we can show that all $(2\ell, \ell, L)$-Furstenberg sets are of size strictly larger than $ q^\ell = n$, we will reach a contradiction as desired.
We note this translation between Furstenberg sets and max load of linear maps 
has been made previously in work of Dhar, Dvir, and Lund \cite[Lemma 15]{ddldrizzy}

Unfortunately, Furstenberg set lower bounds in the literature are extremely poor in this regime of parameters. The main issue is the fact that $T$ points is a vanishingly small fraction of the $q^\ell$ points in a coset. We call this the ``sparse'' regime. Most arguments in the literature go through the polynomial method, which give optimal lower bounds for Kakeya sets. Since Kakeya sets require a shift of every subspace to be fully contained in it, they can be thought of as the ``densest'' instances of Furstenberg sets. The method, as used traditionally used, severely degrades as the Furstenberg instances get sparser. The best known Furstenberg set lower bound in this regime shows $(2\ell, \ell, L)$-Furstenberg sets have size at least $(L/q^\ell)^{2\ell} q^{2\ell}\exp(-{2\ell}/{q^{\ell - 1}}) = q^{-\Omega(\ell^2)}$ \cite[Theorem 6.14]{dhardvir}, which is astonishingly worse than the $q^\ell$ lower bound we desire.

\subsubsection{Strong Furstenberg Set Lower Bounds}
We now explain our approach to proving strong Furstenberg lower bounds in this sparse regime. To do so, we first give a brief overview of the polynomial method of multiplicities technique used by Kopparty, Lev, Saraf, and Sudan~\cite{KLSS} to show optimal Kakeya set lower bounds. Let $S\subset \F_q^n$ be $(2\ell, \ell)$-Kakeya.
Their argument studies a
flavor of algebraic complexity of a Kakeya set
by encoding it as a polynomial
$P \in \F_q[x_1, \dots, x_n]$
with certain constraints
and studying its properties.

In particular, 
let the \emph{multiplicity complexity}
$d_M(S)$ be the smallest possible degree that $P$ can have
such that it vanishes with
multiplicity $M$ on $S$.
That is, if we take any $M-1$
derivatives on $P$,
the resulting polynomial will still vanish
at every point in $S$.
On the one hand,
we can upper bound $d_M(\cdot)$
for \emph{any} set
as a strictly increasing function
of its size:
small sets have smaller $d_M$,
simply by bounding the number of constraints imposed on $P$.
On the other hand,
we can show that $d_M(S)$
is large when $S$ is Kakeya
by utilizing its structure.

Consider a polynomial $P\in \F_q[x_1,\dots, x_n]$ that vanishes on every $x\in S$ with multiplicity $M$. The goal will be to show that $P$ will vanish with high multiplicity over a very large product set over the rational function field $\F_q(t_1,\dots, t_\ell)$. Multiplicity Schwartz--Zippel then converts the large product set size and large total multiplicity into a large degree lower bound on $P$. 

 For any subspace $V\le_\ell \F_q^n$, there exists some shift $b_V\in \F_q^{2\ell}$ such that $b_V+V\subset S$. $V$ can be parametrized by a linear form in $\ell$ formal variables $t_1,\dots, t_\ell$, and thus the restriction of $P$ to $b_V + V$ can be written as a polynomial in $\F_q(t_1,\dots, t_\ell)$ of total multiplicity $Mq^\ell$. One can gather analogous information over all subspaces $V$ and derivatives of $p$, glue them together using multiplicity Schwartz--Zippel, and derive a lower bound on $d_M$, which will scale with the total multiplicity garnered at each restriction.

What happens when $S$ is $(2\ell, \ell, L)$-Furstenberg rather than Kakeya? It is natural to attempt the same argument, where we look at a minimal degree polynomial $P\in \F_q[x_1,\dots, x_n]$ which vanishes on $S$ with multiplicity $M$. The constraint-counting argument remains the same, but the Schwartz--Zippel argument now suffers. For every $V\le_\ell \F_q^n$, we know there is a shift $b_V$ such that $b_V+V$ intersects $S$ in $\ge L$ vectors. Hence, $P$ restricted to this affine space will be an $\ell$-variate polynomial with total multiplicity $M L\ll Mq^\ell$. This multiplicity degradation causes the gluing procedure to give a much smaller degree lower bound. We remark that this approach recovers the Furstenberg lower bounds of Dhar and Dvir mentioned earlier using a different argument \cite[Theorem 6.14]{dhardvir}.

\paragraph{Polynomial method of multiplicity gaps.}
Our idea to avoid this severe multiplicity degradation is surprisingly simple: enforce multiplicity constraints on $x\notin S$ as well. This initially seems like a horrible idea, as this implies our polynomial will now be zero \emph{everywhere}, rather than on $S$ alone.
How can such a polynomial encode any information about $S$? We crucially set the multiplicity constraints on $x\notin S$ to be 
more relaxed than on $x\in S$. Hence, 
information about our set is encoded solely in the derivatives---in the \emph{gap} between the two classes of multiplicity constraints---rather than in the zero set.

More concretely, for parameters $M' < M$, we define
the \emph{multiplicity gap complexity}
$d_{M,M'}(S)$ to be the smallest nonzero degree of a polynomial $P$ which vanishes with multiplicity $\ge M$ on each $x\in S$, and vanishes with multiplicity $\ge M'$ on $x\notin S$. Then for every $V\le_r \F_q^n$, the restriction of $P$ to $b_V + V$ will now have total multiplicity $M L + M'(q^\ell - L)$, which gets very close to the Kakeya-case quantity $Mq^\ell$ as $M'$ approaches $M$.

However, there is a cost to increasing $M'$. The larger $M'$ is, the more linear constraints each $x\notin S$ will induce in the coefficients, which then degrades the dependence of the degree upper bound on $|S|$. By performing the computation, one observes that the number of linear constraints is dominated by enforcing multiplicity $M$ on $S$, unless $M'$ is extremely close to $M$. We thus set $M' = (1-\eps)M$, compute the lower bound implied by the polynomial method, and optimize the value of $\eps$ to yield our optimal bounds.

 To the best of our knowledge, applications of polynomial method to lower bound sets with structure always consider polynomials that vanish on that set (perhaps with high multiplicity), but never enforce any conditions outside of $S$. In the setting of Furstenberg sets, enforcing multiplicity outside of the set allows us to get optimal lower bounds.  We hope that this technique of enforcing multiplicity outside of $S$ will have further applications.

 \paragraph{Furstenberg acrobatics.}
To generalize our results to arbitrarily large universe and to develop tail bounds strong enough to deduce expectation, we require ways to strengthen parameters of Furstenberg sets. 

Given a $(\delta, n, r, T)$-Furstenberg set $S$, where $\delta \ll 1$, we use a classic ``random rotations'' trick to construct a $(1/2, n, r, T)$-Furstenberg set whose size is not much bigger than $S$. This is done by taking unions of sufficiently many random rotations of $S$ to cover more directions, and was most relevantly used by Ordentlich, Regev, and Weiss \cite[Theorem 2.8]{orw22} for Kakeya sets. This allows us to amplify tail bounds.

We also provide a new way to compress the ambient dimension of a Furstenberg set. Given a $(\delta, n, r, T)$-Furstenberg set $S$ and $n'\ll n$, we can construct a $(\delta, n', n'-n+r, T )$-Furstenberg set of size at most $|S|$ . This is done by randomly mapping $S\subset\F_q^n$ into $\F_q^{n'}$ using a random linear map, and showing that for some choice of linear map, the image of $S$ is the desired Furstenberg set. To the best of our knowledge, this maneuver appears to be new, but is heavily inspired by a compression trick done in the work of Alon, Dietzfelbinger, Miltersen, Petrank, and Tardos \cite{admpt} on linear hashing. This trick allows us to generalize from input dimension $2\ell$ to arbitrarily large dimensions.

\subsubsection{List Decodability of RLCs}
After getting max load bounds for linear hashing, list decoding bounds for RLCs immediately follow. This is because a linear code $C\le \F_q^n$ of rate $R$ being $(p,L)$-list decodable is exactly equivalent to $M(h,B)\le L$, where $h:\F_q^{n}\to\F_q^{(1-R)n}$ is any linear map of kernel $C$, and $B\subset \F_q^n$ is Hamming ball of radius $pn$. Therefore, analyzing the list decodability of an RLC is equivalent to analyzing how well a random linear map hashes Hamming balls (rather than arbitrary sets). RLC list decodability is consequently a simple application of our hashing result.

\section{Tight Bounds for Furstenberg Sets}

\subsection{Preliminaries}
\label{sec:preliminaries}

$U \le_r V$ denotes that $U$ is a dimension-$r$ subspace of $V$.

\paragraph{Hasse derivatives.}
We make use of Hasse derivatives, which are
an analog of traditional derivatives which are useful
for multivariate polynomials over finite fields.
However, all of the properties of Hasse derivatives
that we use have analogs
that hold for traditional derivatives over $\R$. These properties and their proofs can be found in the work of Dvir, Kopparty, Saraf, and Sudan \cite{dkss}.

\begin{definition}[Hasse derivative]
  Let $X = (X_1,\dots, X_n)$ and $P(X)\in \F[X]$ be an $n$-variate polynomial. Let $\vec{i}$ be a non-negative integer vector of dimension $n$. The $\vec{i}$-th Hasse derivative of
  $P$, denoted $P^{(\vec{i})}$, is the 
  unique coefficient of $Z^{\vec{i}}$ in the expansion of $P(X+Z)\in \F[X,Z]$:
  \[
    P(X+Z) = \sum_{\vec{i}} P^{(\vec{i})}(X)Z^{\vec{i}}.
  \]
\end{definition}
From the definition, it is clear that the Hasse derivative is linear. 

\begin{fact}
  Let $P$ be a polynomial.
  Then, every coefficient of every Hasse derivative of $P$ is a linear combination of the coefficients of $P$.
\end{fact}

For a vector $\vec{i}\in \Z_{\ge 0}^n$, let $||\vec{i}|| = \sum_{j=1}^n i_j$. Just as in the reals, taking a derivative reduces the degree of a polynomial by 1.

\begin{fact}[{\cite[Lemma 2.1.4]{saraf_2011}}]
  \[\deg(P^{(\vec{i})}) = \deg(P) - \norm{\vec{i}}.\]
\end{fact}

For a polynomial $P\in \F_q[x_1,\dots, x_n]$, denote $P_H$ to be the homogeneous part of $P$, i.e., the unique homogeneous polynomial such that $\deg(P-P_H) < \deg P$. The derivative of the homogeneous part is the homogeneous part of the derivative.

\begin{fact}[{\cite[Proposition 2.3]{dkss}}]
\label{fact:hasse-homogeneous}
\[(P^{(\vec{i})})_H \equiv (P_H)^{(\vec{i})}\]
\end{fact}

We now define the multiplicity of a polynomial at a value.

\begin{definition}[Multiplicity]
  Let $P\in \F_q[x_1,\dots, x_n]$ and $a\in \F_q^n$. We define $\mult(P, a)$ to be the largest nonnegative integer $m$ such that for all $\vec{i}\in \Z_{\ge 0}^n$, with $||i||< m$, $P^{(\vec{i})}(a) = 0$.
\end{definition}

If a polynomial has a zero with multiplicity at least $m$ at some point, the derivative of the polynomial will have mulitplicity at least $m-1$ at that point.

\begin{fact}[Multiplicity loss {\cite[Proposition 2.3]{dkss}}]
  \label{fact:multiplicity-loss}
  \[
    \mult(P^{(\vec{i})}, \vec{a}) \ge \mult(P, \vec{a}) - \norm{\vec{i}}.
  \]
\end{fact}

We finally present a very important lemma, relating the number of zeros (with multiplity) of a polynomial to its degree.

\begin{lemma}[Multiplicity Schwartz--Zippel, {\cite[Lemma 2.7]{dkss}}]
  \label{lem:multiplicity-sz}
  If $P$ is a $n$-variate nonzero\footnote{A polynomial that is not
    the literal zero---$P$ might still evaluate to zero everywhere!}
  polynomial over $\F$
  of (total) degree at most $d$, then
  for all $S \subseteq \F$,
  \[
    \sum_{\vec{x} \in S^n} \mult(P, \vec{x}) \le d \cdot \abs{S}^{n-1}.
  \]
\end{lemma}

\subsection{Method of Multiplicity Gaps}
Let us begin with the a lower bound which holds for every value of a parameter $\eps$,
to pick later.
\begin{theorem}
\label{thm:technical-main}
  \label{thm:furstenberg-bound}
  Let $S \subset \F_q^n$ be $(\delta, r, \alpha q^r)$-Furstenberg.
  For all $0 < \eps \le 1$ and
  $\beta_\eps \coloneq \alpha + (1-\eps)(1-\alpha)$,
  \begin{equation}
    \abs{S}
    \ge q^n \cdot \frac{ \left(\frac{\delta \beta_\eps q^r }{\delta q^r+q-1}\right)^n -
      (1-\epsilon)^n}{1-(1-\epsilon)^n},
    \end{equation}
  or equivalently
  \begin{align}
    \abs{\conj{S}}
    &\le q^n \cdot \frac{ 1-\left(\frac{\delta
          \beta_\eps q^r}{\delta q^r+q-1}\right)^n}{1-(1-\epsilon)^n}\\
    &= q^n\cdot  \frac{ 1-\left(1 - \frac{\delta \eps(1-\alpha) q^r +q-1}{\delta
          q^r+q-1}\right)^n}{1-(1-\eps)^n}.
      \end{align}
\end{theorem}

We prove this using the polynomial method of \emph{multiplicity gaps},
building on the extended multiplicity method of \cite{dkss, KLSS}
with a new twist.
At a high level,
we will analyze the
\emph{multiplicity gap complexity}
of a set $S$:
what is the smallest $d$
such that there exists a nonzero degree-$d$ polynomial $P$
that encodes $S$ in its multiplicity gaps?
That is, for some parameter $M$,
we want a $P$
with multiplicity $M$ on every point in $S$
and multiplicity $(1-\eps)M$ on every point in $\conj{S}$.

On the one hand, for \emph{any} set $S$,
there is an easy upper
bound on the multiplicity gap complexity
that depends on the size of $S$,
by counting the linear constraints
imposed on $P$
by the multiplicity requirements.
The smaller the set,
the fewer the imposed constraints
and so the
smaller the multiplicity gap complexity.

On the other hand,
we can lower bound the complexity
when $S$ is Furstenberg.
Because $S$ hits many affine subspaces,
restrictions of $P$ to those subspaces
must vanish with high multiplicity.
All of these subspaces can be treated
as distinct roots
when we work over a function field,
and by Schwartz--Zippel
(in particular, the strengthened variant using multiplicities),
only a high degree polynomial
can have this many roots.
Therefore, a Furstenberg set has high multiplicity gap complexity, so it cannot be small.

\begin{lemma}[Constraint counting {\cite[Lemma 12]{ddldrizzy}}]
\label{lem:constraint-counting}
  For any integer parameter $M$,
  subset $S \subseteq \F^n$, and real $0 < \eps \le 1$,
  choose the smallest $d$ such that
  \begin{equation}
    \label{eq:cc}\binom{M+n-1}n |S| + \binom{(1-\eps) M + n - 1}n|\conj{S}|
    \le \binom{d+n}n.
  \end{equation}
  Then, there
  exists a nonzero $n$-variate polynomial $P$ of degree $d$ over $\F$
  such that
  \begin{equation}
    \label{eq:MC}\tag{M}
  \begin{gathered}
    \forall \vec{x} \in S, \mult(P, \vec{x}) \ge M,\\
    \forall \vec{x} \in \conj{S}, \mult(P, \vec{x}) \ge (1-\eps) M.
  \end{gathered}
  \end{equation}
\end{lemma}
\begin{proof}
  If $\mult(P, \vec{x}) \ge M$ for a point $\vec{x} \in S$,
  then for every $\vec{i}$ with $\norm{\vec{i}} < M$, $P^{(\vec{i})}(\vec{x})=0$.
  The coefficients of $P^{(\vec{i})}$
  are linear combinations of
  coefficients of $P$,
  so this enforces
  a linear constraint on
  the coefficients of $P$.
  Using stars and bars, there are
  $\binom{M+n-1}{n} \abs{S}$ many pairs $(\vec{i}, \vec{x})$ with $\vec{x} \in S$.
  Thus, the total rank
  of the linear constraints imposed on $P$ by $S$ and $\conj{S}$
  is at most the LHS of \eqref{eq:cc}.
  Meanwhile,
  the RHS is the number of coefficient variables in a degree-$d$ $n$-variate polynomial, also by stars and bars.
  Since the constraints do not outnumber the variables of this system,
  there must exist a nontrivial solution,
  i.e., a nonzero $P$
  satisfying the multiplicity constraints.
\end{proof}
\begin{fact}[Multiplicity gap complexity upper bound]
  The degree $d$ chosen in \cref{lem:constraint-counting} satisfies
  \begin{equation}
    \label{eq:degree-upper-bound}
    \binom{M+n-1}n |S| + \binom{(1-\eps) M + n - 1}n|\conj{S}|
    > \binom{d+n-1}n,
  \end{equation}
  because \eqref{eq:cc}
  is strictly increasing in $d$.
\end{fact}

\begin{proposition}[Multiplicity gap complexity lower bound]
  For all $0 < \eps \le 1$ and $\beta_\eps \coloneq \alpha + (1-\eps)(1-\alpha)$,
  let $S \subseteq \F_q^n$ be $(\delta, r, \alpha q^r)$-Furstenberg,
  and let $P$ be nonzero and degree-$d$
  satisfying \eqref{eq:MC}.
  Then,
  \begin{equation}
  \label{eq:degree-lower-bound}
    d > \delta q^r \ceil*{\frac{q \beta_\eps M - d}{q-1}}.
  \end{equation}
\end{proposition}

\begin{proof}
  We will lower bound
  the multiplicity of the homogeneous part of $P$
  on different subspaces
  and
  use multiplicity Schwartz--Zippel
  to get a degree lower bound.
  The first step is to show that on many affine subspaces,
  many
  derivatives of $P$
  vanish.
  \begin{claim}
    For a $\delta$ fraction of $n \times r$ matrices $V$,
    there exists a shift $\vec{b}$
    such that
    for every $\vec{i}$ with $w \coloneq \norm{\vec{i}}$,
    the $r$-variate polynomial
    $Q_{\vi, V}(\vec{t})
      \coloneq P^{(\vec{i})}(\vec{b} + V\vec{t})$
    has total multiplicity at least
    \begin{equation}
      \label{eq:3}
      \sum_{\vec{t} \in \F_q^r} \mult(Q_{\vi,V}, \vec{t})
      \ge \alpha q^r (M - w) + (1-\alpha) q^r ((1-\eps) M - w)
      = q^r (\beta_\eps M - w).
    \end{equation}
  \end{claim}
  \begin{proof}[Claim proof]
  First, consider a random $V$ with full column rank.
  Its columns span a uniformly random
  $r$-dimension subspace,
  so with probability $\delta$,
  there must be a shift $\vb + \colSpan(V)$
  that intersects $S$ in at least $\alpha q^r$ points.
  Every $\vt$ is mapped to a unique point in this shift,
  and using \cref{fact:multiplicity-loss}
  with the multiplicities enforced by
  \eqref{eq:MC} gives \eqref{eq:3}.

  Now, let us extend this to
  {general} $n \times r$ matrices.
  For each $s < r$, take a
  uniformly random rank-$s$ matrix $V$ and
  let $\Lambda \coloneq \colSpan(V)$,
  and complete $\Lambda$
  to a random rank-$r$ subspace $\Lambda'$.
  With probability $\delta$ over $V$ and $\Lambda'$,
  there is a good shift $\vec{y} + \Lambda'$
  that intersects $S$
  at $\alpha q^r$ points.
  This means that for at
  least a $\delta$ fraction of $V$,
  there exists a suitable completion $\Lambda'$
  with such a good shift.

  Consider such $V$ and
  completely tile the good completion
  $\vec{y} + \Lambda'$
  with $q^{r-s}$ shifts of $\Lambda$.
  By averaging, at least one of these shifts
  $\vb + \Lambda$
  hits $S$ at $\alpha q^s$ points.
  Finally, notice that for every $\vec{x} \in \Lambda$,
  there are $q^{r-s}$
  solutions $\vt$ to $V\vt = \vec{x}$.
  Thus, for at least $\delta$ fraction of $V$, there exists $\vec{b}$ so that
  \begin{align*}
    \sum_{\vt \in \F_q^r} \mult(Q_{\vi,V}, \vt)
    &= \sum_{\vec{x} \in \Lambda} q^{r-s} \mult(P^{(\vec{i})}, \vb + \vec{x})
      \geq q^{r-s} \cdot q^s (AM - w),
  \end{align*}
  by \cref{fact:multiplicity-loss}
  and the multiplicity constraints \eqref{eq:MC}.
\end{proof}

Using the contrapositive of multiplicity Schwartz--Zippel,
if $Q_{\vi, V}$ satisfies \eqref{eq:3}
and
\begin{equation}
    q(\beta_\eps M - w) > d - w \ge \deg(Q_{\vi,V}),
\end{equation}
then $Q_{\vi,V}$ is the zero polynomial.
This is true whenever
\[w < m \coloneq \ceil*{\frac{q \beta_\eps M - d}{q-1}}.\]
If $Q_{\vi,V}
\equiv 0$,
we can
show that part of $P^{(\vec{i})}$
must also be zero.
\begin{claim}[Unshifting]
  If $Q_{\vi,V}$ is zero,
  then so is $(P^{(\vec{i})})_H(V\vec{t})$.
\end{claim}
\begin{proof}[Claim proof]
  Recall that $Q_{\vi,V}(\vec{t}) \coloneq P^{(\vec{i})}(\vb + V\vec{t})$.
  Plug $\vec{b} + V\vec{t}$ into each monomial of $(P^{(\vec{i})})_H$
  and expand;
  each such monomial (say of degree $\ell$)
  will expand to
  \[
    (\vb + V\vt)_{j_1}(\vb + V\vt)_{j_2} \dots (\vb + V\vt)_{j_{\ell}}
    = (V\vt)_{j_1}(V\vt)_{j_2} \dots (V\vt)_{j_\ell}
    + (\text{products of $<\ell$ terms}).\]
  The first term on the RHS
  is the same monomial but with $V\vt$ substituted
  instead of $\vb+V\vt$,
  so the sum over these terms across monomials
  will appear in both $Q_{\vi,V}$ and $(P^{(\vec{i})})_H(V\vt)$.
  This is the only possible source of degree-$\ell$
  monomials in $Q_{\vi,V}$,
  so if $(P^{(\vec{i})})_H(V\vt)$ is nonzero,
  then so is $Q_{\vi,V}$.
\end{proof}

We have shown so far that for a $\delta$ fraction of $n \times r$ matrices $V$,
every $\norm{\vec{i}} < m$
will give $(P^{(\vec{i})})_H(V\vt) \equiv 0$
(equivalently, $(P_H)^{(\vec{i})}(V\vt) \equiv 0$ via \cref{fact:hasse-homogeneous}).
Let us now treat $V\vt$ itself as the parameter
by viewing $(P_H)^{(\vec{i})}$
as a polynomial over the $r$-variate {function field}
$\F_q(\vt)$---now,
$(P_H)^{(\vi)}(V\vt) \equiv 0$
means that $V\vt$
itself is a root.
Define the set of $n$-output linear functions on $r$ variables,
which is a product set:
\[
T^n \coloneq \{V\vt : V \in \F_q^{n \times r}\}, \quad
T \coloneq \{\angles{\vec{u}, \vt} : \vec{u} \in \F_q^r\} \subseteq \F_q(\vt).
\]
On a $\delta$ fraction of $T^n$,
for all $\vi$ with $\norm{\vec{i}} < m$,
$(P_H)^{(\vec{i})}$
is the zero polynomial, which is the zero element
of the function field.
This is saying that
$P_H$ (as a polynomial over $\F_q(\vt)$)
has multiplicity $m$
on a $\delta$ fraction of $T^n$.
$P_H$ overall is nonzero so
by multiplicity Schwartz--Zippel over $\F_q(\vt)$,
\[d \ge \delta m \abs{T} >
\delta q^r\left\lceil\frac{q\beta_\eps M-d}{q-1}\right\rceil. \qedhere \]
\end{proof}

At last, let us prove \cref{thm:furstenberg-bound}
by combining our bounds for $d$.

\begin{proof}[Proof of \cref{thm:furstenberg-bound}]
As a reminder, $M$ is a parameter that we can set to be whatever we want.
With that in mind,
let us compare our upper and lower bounds on $d$ which hold for all $M$.
\eqref{eq:degree-upper-bound} gives
\begin{align*}
\binom{d+n-1}n
&< \binom{M+n-1}n \abs{S} + \binom{(1-\eps) M + n - 1}n\abs{\conj{S}} \\
&=  \parens*{\binom{M+n-1}n-\binom{(1-\eps) M + n - 1}n} \cdot \abs{S} + q^n\binom{(1-\eps) M+n-1}n.
\end{align*}
Rearranging this gives
a lower bound for $\abs{S}$
for every choice of $M$, including in the limit $M \to \infty$:
\begin{gather}
\label{eq:eta-tau-bound}
\abs{S}
\ge \frac{\eta - q^n \tau}
{1-\tau},\\
\eta \coloneq \frac{\binom{d+n-1}n}{\binom{M+n-1}n},\quad
\tau \coloneq \frac{\binom{(1-\eps) M+n-1}n}{\binom{M+n-1}n}.
\end{gather}
We can analyze the behavior of $\eta$ and $\tau$ as follows.
\begin{gather}
\label{eq:tau-limit}
\tau = \frac{(1-\eps) M + n - 1}{M + n - 1} \cdot
\frac{(1-\eps) M + n - 2}{M + n - 2}
\cdots
\frac{(1-\eps) M}{M}
\implies \lim_{M \to \infty} \tau = (1-\eps)^n.\\
\label{eq:eta-lower-bound}
\eta = \frac{d + n - 1}{M + n - 1} \cdot
\frac{d + n - 2}{M + n - 2}
\cdots
\frac{d}{M} \ge \parens*{\frac{d}{M}}^n.
\end{gather}
\eqref{eq:degree-lower-bound} gives a lower bound
for $d/M$:
\begin{equation}
\label{eq:ratio-lower-bound}
    \frac{d}{M} > \frac{\delta \beta_\eps q^{r+1}}{\delta q^r + q - 1}.
\end{equation}
Plugging %
into \eqref{eq:eta-tau-bound}
and taking the limit $M \to \infty$,
we get our bound.
\begin{align*}
\abs{S}
&\ge \lim_{M \to \infty} \frac{\eta - q^n \tau}
{1-\tau}\\
&\ge q^n\cdot\frac{\parens*{\frac{\delta \beta_\eps q^{r}}{\delta q^r+q-1}}^n -(1-\eps)^n}{1-(1-\eps)^n}.\qedhere
\end{align*}
\end{proof}

Next,
we choose $\eps$
to get a good lower bound.
This bound is nearly tight;
see \cref{apx:tight}.

\begin{corollary}
\label{cor:furst-constant-density}
    Let $S\subset \F_q^n$ be $(\delta, n, r, \alpha q^r)$-Furstenberg.
    For all $1 < C < \alpha \delta q^{r-1}/2$,
    \[
    \frac{\abs{S}}{q^n} \ge \exp\parens*{-\frac{2Cn}{\delta \alpha q^{r-1}}}
    \cdot \parens*{1 - \frac{1}{C}}
    \cdot \alpha.
    \]
\end{corollary}

\begin{proof}
    Set $\gamma \coloneq \frac{q-1}{\delta q^r}$. From \cref{thm:technical-main}, we have for every $\eps\in (0,1)$ that \begin{align*}\frac{|S|}{q^n} & \ge \frac{\left(1 - \frac{\delta\eps(1-\alpha)q^r + q - 1}{\delta q^r + q-1}\right)^n-(1-\eps)^n}{1-(1-\eps)^n} = \frac{\left(1 - \frac{\eps(1-\alpha) +\gamma}{1+\gamma}\right)^n-(1-\eps)^n}{1-(1-\eps)^n}. \end{align*}

    By the mean value theorem, we have that for $x > y > 0$, $x^n-y^n\ge ny^{n-1}(x-y)$, and by Bernoulli's inequality, we have $1-(1-\eps)^n\le n\eps$. Hence, 
\begin{align}
\frac{|S|}{q^n} & \ge \frac{n(1-\eps)^{n-1}\left(\eps - \frac{\eps(1-\alpha) +\gamma}{1+\gamma}\right)}{n\eps}  = \frac{\alpha + \gamma -  \gamma/\eps}{1+\gamma}\cdot\left(1 - \eps\right)^{n-1}.\label{eq:MVT}
\end{align}
As $\eps$ gets smaller, the second factor becomes larger, but the first factor decreases, and even becomes negative once $\eps \le {\gamma}/({\alpha + \gamma})$. However, we note that for any $\eps \ge {C\gamma}/({\alpha + \gamma})$ that is growing, the first term remains the same asymptotically, while the second term decays. This hints that we should set $\eps = {C\gamma}/({\alpha + \gamma})$ for some constant-bounded $C > 1$. Plugging this in, we have

\begin{align}
     \left(\frac{\alpha + \gamma -  \gamma/\eps}{1+\gamma}\right)\left(1 - \eps\right)^{n-1} & = \left(1-\frac{1}C\right)\frac{\alpha + \gamma}{1+\gamma}\left(1 - \frac{C\gamma}{\alpha + \gamma}\right)^{n-1} \\ & \ge \left(1-\frac{1}C\right)\alpha\left(1 - \frac{C\gamma}{\alpha}\right)^{n-1}.
\end{align}
Assuming $1 - {C\gamma}/\alpha \ge 1/2$, %
we use $1 - x \ge \exp(-x/(1-x))$ to obtain \[\frac{|S|}{q^n} \ge \left(1-\frac{1}C\right)\alpha\exp\parens*{-\frac{2Cn\gamma}{\alpha}} \ge \left(1-\frac{1}C\right)\alpha\exp\parens*{-\frac{2Cn}{\alpha\delta q^{r-1}}}.\qedhere\]
\end{proof}

\begin{remark}
    The constant $2$ inside 
    the exponential term in
    \cref{cor:furst-constant-density},
    which will affect the load bound
    for hashing
    and the list size
    for random linear codes,
    can be set arbitrarily close to $1$
    as long as the correponding
    assumption on $C$ is updated.
    For clarity,
    we will not optimize this and other
    constants throughout.
\end{remark}

\subsection{Furstenberg Acrobatics}
\subsubsection{Random Rotations}
\label{sec:random-rotations}

 We now improve the Furstenberg set lower bound for small $\delta$ by using a ``random rotations'' argument, just like how Ordentlich, Regev, and Weiss did for Kakeya sets \cite[Lemma 2.6]{orw22}. This is crucial to getting tail bounds sharp enough to obtain optimal expectation and also to get optimal probability of achieving list decoding capacity.

\begin{lemma}
\label{lem:random-rotation}
    Let $0 < \eps < \delta < 1$. Assume $S\subset \F_q^n$ is $(\eps,n, r, T)$-Furstenberg. There exists $S'$ which is $(\delta, n, r, T)$-Furstenberg with
    \begin{equation}
        \label{eq:random-rotations}\abs{S'} \le \ceil*{\frac{\log(1-\delta)}{\log(1-\eps)}} \cdot \abs{S}.
    \end{equation}
\end{lemma}

\begin{proof}
    The idea will be to take the union of random rotations of $S$. These random rotations will introduce new ``rich'' directions, thereby boosting $\eps$ to $\delta$.
    We will then upper bound the number of random rotations needed, giving the upper bound on the size of $S'$.

    Let $M_1,\dots, M_k$ be random and independent $n\times n$ invertible matrices over $\F_q^n$. Consider $S' = \bigcup_{i\in [k]} M_iS$, where $MS = \{Ms: s\in S\}$. Fix an $r$-dimensional subspace $V\le \F_q^n$. For a fixed $i$, the probability that $V$ is not a rich direction
    in $M_iS$
    (does not intersect $M_i S$ in at least $\alpha q^r$ points)
    is at most $1-\eps$. This is because at least $\eps$ fraction of subspaces are rich for $S$, and an $M_i$ transitively permutes these rich subspaces. Hence the probability $V$ is not rich when we take the union of $k$ independent rotations is $(1-\eps)^k$. Thus, there exists a choice of $M_1,\dots, M_k$ such that the number of rich directions of $S'$ comprise of at least $1-(1-\eps)^k$. For this quantity to be at least $\delta$, we need $k = \ceil{ {\log(1-\delta)}/{\log(1-\eps)} }$. The result now follows by noting $|S'|\le k|S|$.
\end{proof}

Combine this random rotations lemma with our prior lower bound \cref{cor:furst-constant-density}.
This shifts the dependence on $\delta$
from inside the exponential to outside.

\begin{corollary}
\label{cor:furst-improved-delta}
    Let $S\subset \F_q^n$ be $(\delta,n, r, T)$-Furstenberg.
    As long as $T \ge 8q$,
    \[
    {\abs{S}}\ge
    \frac{\delta T}{4} \cdot  q^{n-r} \exp\parens*{-\frac{8qn}{T}}\]
\end{corollary}

\begin{proof}
We will use the following claim:
\begin{claim}
For $x \in (0,1)$,
    \[
    \ceil*{\frac{\log(1/2)}{\log(1-x)}}
    \le \frac 2 x.
    \]
\end{claim}
\begin{proof}[Claim proof]
First,
\begin{align}
    \ln(1-x) \le -x \implies
    \frac{\ln(1/2)}{\ln(1-x)}
    \le \frac{\ln 2}{x}.
\end{align}
Then,
\begin{align}
    \ceil*{\frac{\log(1/2)}{\log(1-x)}}
    &\le
    {\frac{\log(1/2)}{\log(1-x)}} + 1\\
    &\le \frac{x + \ln 2}{x}\le \frac 2 x,
\end{align}
since $x < 1 < 2 - \ln 2$.
\end{proof}
Now, assume for the sake of contradiction that there exists set $S$ contradicting the corollary. By \cref{lem:random-rotation}, there exists a $(1/2,n, r, T)$-Furstenberg set $S'$ of size
\begin{align}
    {|S'|}
    &< \ceil*{\frac{\log(1/2)}{\log(1-\delta)}}\cdot
    \frac{\delta T }{4} \cdot q^{n-r} \exp\parens*{-\frac{8qn}{T}}\\
    &\le
    \frac{T}{2} \cdot q^{n-r} \exp\parens*{-\frac{8qn}{T}}.
\end{align}
However, this directly contradicts \cref{cor:furst-constant-density}
with $C=2$.
\end{proof}

\subsubsection{Random Projections}
Jumping ahead a bit,
\cref{cor:furst-improved-delta}
is strong enough
to show optimal expected max load 
for hash functions
on ambient dimension $n = {O(\ell)}$,
which is the expository case
described in \cref{sec:proof-overview},
as well as
optimal list size for random linear codes
with rate bounded away from zero.
Unfortunately, this lower bound is meaningless
when $n \gg T/q$,
which is the setting of hashing from large universes
or random linear codes with vanishing rate.

We will use a dimension reduction idea
to improve the 
Furstenberg set lower bound
in the case where $r = n-k$
for small $k$,
at the cost of 
a leading constant factor.
We can do so by taking
a random projection of 
a $(\delta, n, n-k, T)$-Furstenberg set
which reduces the ambient dimension
drastically to something which 
only depends on $k$ and $T$. This is inspired by a compression step done in \cite{admpt}.
The first step
is to show that a random linear map
that compresses
\emph{any} arbitrarily large vector space
to $2t$ dimensions
will be surjective on
an arbitrary fixed set of size $q^t$
with constant probability.
This is optimal
as it is the converse of
the birthday paradox.

\begin{lemma}[Dimension reduction]
\label{lem:collision}
    Let $V$ be a vector space, and let $f\colon V\to \F_q^{2t}$ be a uniformly random linear map. For every set $B\subseteq V$ of size at most $q^t$, \[\Pr_f\braks*{\abs{f(B)}
    = \abs{B}}
    \ge \frac{1}2.\]
\end{lemma}
\begin{proof}
    By 2-universality,
    any fixed pair of vectors collide with probability $1/q^{2t}$.
    There are $\binom{\abs{B}}2$ pairs of vectors from $B$, so the expected number of collisions within $B$ is at most \[q^{-2t}\binom{\abs{B}}2 \le \frac{\abs{B}^2}{2q^{2t}}\le 1/2.\]
    Therefore, with probability $\ge 1/2$, there are no collisions, which implies $\abs{f(B)} = \abs{B}$.\qedhere
\end{proof}

If we take a $(\delta, n, n-k, T)$-Furstenberg set $S$ over $\F_q^n$
and apply a random projection,
the resulting set
will still be meaningfully Furstenberg,
because this lemma shows that
many affine subspaces that 
intersect $S$ at $T$ points
will continue to do so after projection.
But our new Furstenberg set is over a  significantly smaller ambient dimension,
so we can apply our Furstenberg lower bound
which implies a much stronger lower bound
for $S$.

\begin{theorem}[formal version of \cref{thm:intro-main}]
\label{thm:furst-dim-reduce}
Let $\delta > 0$, 
$n\ge k\ge 1$, $q$ be a prime power, $8q\le T\le q^{n-k}$, and $t = \ceil{\log_q T}$. 
Suppose $S$ is $(\delta, n, n-k, T)$-Furstenberg.
Then, there exists a $(\delta/2, k+2t, 2t, T)$-Furstenberg set of size at most $|S|$,
which implies
\[
\abs{S} \ge \frac{\delta T}{8}
\cdot q^k 
\exp\parens*{-\frac{8q(k+2t)}{T}}.
\]
\end{theorem}
\begin{proof}

    Let  $h: \F_q^n\to\F_q^{k+2t}$ be a uniformly random \emph{surjective} linear map.
    Consider a
    uniformly random subspace $W\le_{2t}\F_q^{k+2t}$. We will show that 
    \begin{equation}\label{eq:pr-rich}\Pr_{h,W}[\exists b\in \F_q^{k+2t}, \abs{(b + W)\cap h(S)} \ge T]\ge \frac{\delta}2.\end{equation} Then it will follow that the proportion of directions $W\le_{2t}\F_q^{k+2t}$ which have a shift that intersects $h(S)$ on at least $T$ points is at least $\delta/2$ in expectation. We can then fix $h$ that achieves at least this average, and obtain that $h(S)$ is $(\delta/2, k+2t, 2t, T)$-Furstenberg of size $\abs{h(S)}\le \abs{S}$.
    Apply \cref{cor:furst-improved-delta}
    to get a lower bound.

    Define the random subspace $V\coloneqq h^{-1}(W)$. We note that $V$ will be a random subspace of dimension $n-k$ in $\F_q^n$. Define the event \[{\cal E} \coloneq 
    (\exists b_V\in \F_q^n, \abs{(b_V+V)\cap S}\ge T).\] By definition of $S$, we know that $\Pr_{h,W}[{\cal E}]\ge \delta$. Condition on $\mathcal{E}$
    which fixes some subspace $V$ which is heavy for $S$,
    and let $S_V$ be an arbitrary subset of size $T$ from $(b_V+V)\cap S$. 
    Despite this conditioning, $g \coloneq h\rvert_V\colon V\to W$ is still a uniformly random surjective map.
    And because $S_V - b_V\subseteq V$ is a set of size $T\le q^t$,
    \cref{lem:collision} implies that
    $g$ will preserve its size often:
    \[\Pr[\abs{g(S_V - b_V)} = T \mid \mathcal{E}]\ge 1/2.\]
    But this implies that
    \[T = |g(S_V - b_V)| \le |h((S-b_V)\cap V))| =  |h(S-b_V)\cap W| = |h(S)\cap (h(b_V)+W)|,\] and thus there is a shift of $W$ that intersects $h(S)$ in at least $T$ places. Hence, \[\Pr_h[\exists b\in \F_q^{k+2t}, |(b + W)\cap h(S)|\ge T]\ge \Pr[{\cal E}] \cdot \Pr[|g(S_V - b_V)| = T \mid {\cal E}] \ge \frac{\delta}2,\] establishing \eqref{eq:pr-rich}.
\end{proof}

\section{Linear Hashing over \texorpdfstring{$\bm{\F_q}$}{Fq}}

We will explicitly denote 
dependence on $q$
anywhere it appears.
\begin{definition}[Max load]
    For a function $h: A\to B$ and subset $S\subset A$, we define \[M(S,h)\coloneqq \max_{b\in B} |h^{-1}(b)\cap S|\] to be the \emph{max load} when $h$ is viewed as a hash function mapping $S\subset A$ into $B$.
\end{definition}  

\subsection{Max Load Tail Bounds}
A tail bound for the max load of a 
random (surjective) linear hash function
is a special case of
a Furstenberg lower bound.
The contents of each bin
are determined by shifts
of the kernel of the hash function,
so a bad set of balls $S$
is precisely one where many subspaces
have shifts that heavily intersect $S$.
If we directly apply our Furstenberg
lower bound to the hashing setting,
we get the following tail bound:

\begin{theorem}[Load tail bound from Furstenberg]
\label{thm:mn-tail-direct}
Let $q$ be a prime power, $u\ge\ell\ge 1$ be integers,
and $h:\F_q^u\to\F_q^\ell$ be a uniformly random \emph{surjective} $\F_q$-linear map.
For any $T \ge 8q$ and
any $S\subseteq \F_q^u$,
\[\Pr\braks*{M(S,h)\ge T} \leq
  \frac{8\abs{S}}{q^\ell T}  \exp\parens*{\frac{8q(\ell+2\ceil{\log_qT})}{T}}.\]
\end{theorem}

\begin{proof}
Denote $\delta \coloneqq \Pr\left[M(S,h)\ge T\right]$ and $t\coloneqq \ceil{\log_q T}$.
The kernel of $h$ is a uniformly random subspace of $\F_q^{u}$ of dimension $u-\ell$. If $M(S,h)\ge T$, then some shift of this kernel intersects $S$ in at least $T$ points. Therefore, $\Pr_h[M(S,h)\ge T] = \delta$ implies that 
$S$ is $(\delta, u, u-\ell, T)$-Furstenberg. By \cref{thm:furst-dim-reduce}, this implies that 
\begin{gather*}
|S|
\ge \frac{\delta T}{8} \cdot q^\ell\exp\parens*{-\frac{q(\ell + 2t)}{T}}\\
     \implies
     \delta \le \frac{8\abs{S}}{q^\ell T} \exp\parens*{\frac{q(\ell+2t)}{T}}.\qedhere
   \end{gather*}
\end{proof}

For linear hashing, we are interested in a uniform random linear map, rather than a uniform surjective linear map.
We can generalize the tail bound by showing that linear maps are surjective with high probability. 
If $u$ is sufficiently large,
the probability that $h$ is not surjective is small enough that we can union bound it with the tail.

\begin{theorem}[General tail bound for large $u$]
\label{thm:remove-surject-large-u}
Let $q$ be a prime power, $u\ge\ell\ge 1$ be integers,
and $h:\F_q^u\to\F_q^\ell$ be a uniformly random $\F_q$-linear map.
In addition,
suppose that $u$
is sufficiently large,
i.e.,
$u\ge u_0(\ell, |S|, T,q)$.
Then for any $T \ge 8q$ and
any $S\subseteq \F_q^u$, we have
\[\Pr\braks*{M(S,h)\ge T} \leq
  \frac{9\abs{S}}{q^\ell T}  \exp\parens*{\frac{8q(\ell+2\ceil{\log_qT})}{T}}.\]
\end{theorem}

\begin{proof}
When $u$ is sufficiently large
relative to the other parameters,
then the probability that $h$
is not surjective is tiny.
Thus, the max load tail bound
conditioned on surjectiveness
extends to this general case
with a tiny loss.
Let $\cal E$ be the probability that $h$ is not surjective. Viewing $h$ as an $\F_q^{\ell \times u}$ matrix, $\cal E$ is equivalent to the probability this matrix is not full rank, which is \[1-(1-q^{-u})(1-q^{-u+1})\cdots(1-q^{-u+\ell - 1})\le \sum_{i=0}^{\ell - 1} q^{-u + i}\le q^{-u+\ell}.\]

Conditioned on $\neg {\cal E}$, $h$ is a random linear surjective map for which \cref{thm:mn-tail-direct} applies. Hence, 
\begin{align*}\Pr[M(S,h)\ge T] & \le \Pr[{\cal E}] + \Pr[M(S,h)\ge T \mid {\lnot\cal E}] \\ & \le q^{-u+\ell} + \frac{8\abs{S}}{q^\ell T}  \exp\parens*{\frac{8q(\ell+2\ceil{\log_qT})}{T}} \\ & \le \frac{9\abs{S}}{q^\ell T}  \exp\parens*{\frac{8q(\ell+2\ceil{\log_qT})}{T}}\end{align*}
as long as $u_0(\ell, \abs{S}, T, q)$
is set large enough.
\end{proof}

Finally, we can extend this tail bound to every $u$ with an appropriate embedding.
If the tail bound holds over
large $u' > u$,
we can lift our set $S$
and our random hash function $h$
from $\F_q^u$ to $\F_q^{u'}$
while preserving
which balls go into which bins
(and hence preserving the max load).
\cref{thm:remove-surject-large-u} then applies, and this bound
has no dependence on $u'$ so 
it also applies to our original setting.
\begin{corollary}[General tail bound for all $u$]
\label{thm:remove-surject}
Let $q$ be a prime power, $u\ge\ell\ge 1$ be integers,
and $h:\F_q^u\to\F_q^\ell$ be a uniformly random $\F_q$-linear map.
Then for any $T \ge 8q$ and
any $S\subseteq \F_q^u$, we have
\[\Pr\braks*{M(S,h)\ge T} \leq
  \frac{9\abs{S}}{q^\ell T}  \exp\parens*{\frac{8q(\ell+2\ceil{\log_qT})}{T}}.\]
    
\end{corollary}
\begin{proof}
If $u \ge u_0(\ell,\abs{S},T,q)$
from \cref{thm:remove-surject-large-u},
then we are done.
Otherwise, pick $u' = u_0(\ell, \abs{S}, T, q)$,
    and lift $S$ and $h$ to $\F_q^{u'}$:
    add zeroes to the end of each tuple in $S$,
    and add uniformly random columns
    to the matrix of $h$,
    turning it into
    a uniformly random linear map
    on a larger universe.
\Cref{thm:remove-surject-large-u}
gives a tail bound for the larger universe,
but the bin assignments
(and hence max load)
are the same as if we
had evaluated
$h$ on $S$
in the original, smaller universe.
\end{proof}

\subsection{Optimal Expected Max Load}

Combining our tail bounds with tail bounds of prior art \cite{admpt} allows us to show that linear hashing over $\F_q$ achieve optimal expected max load for any constant $q$.

\begin{theorem}[formal \cref{thm:intro-eml}]
\label{thm:eml-final}
There is an absolute constant $\ell_0$
such that for all integers
$\ell \ge \ell_0$, the following holds. For \emph{any} prime power $q$,
integer $u\ge \ell$,  subset $S\subset \F_2^u$ of size $n \coloneq q^\ell$, and uniformly random linear map $h:\F_q^u\to\F_q^\ell$,
we have \[\E_h [M(S,h)]\le {\frac{10q\ell}{\ln\ell}}
= O\parens*{\frac{q \ln \ln q}{\ln q} \cdot {\frac{\ln n}{\ln \ln n}}}.\]
\end{theorem}

\begin{proof}
   We will write the expectation as the integral of the upper-tail bound.
   \begin{align}
       \E_h[M(S,h)] & = \int_0^\infty \Pr[M(S,h)\ge t]\,dt.
   \end{align}
    Evaluate this integral by splitting it into three regions:
    \begin{enumerate}[(1)]
        \item
        In the first region when $t$ is small, we expect $\Pr[M(S,h)\ge
        t] = \Omega(1)$. We will thus use a trivial bound of 1 until our
        tail bound starts to predict tail decay.
        Let us integrate up to $t\le 9q\ell/\ln \ell$.
        This term will dominate the integral.
        \begin{equation}
            \int_0^{9q\ell/{\ln\ell}} \Pr[M(S,h)\ge t]\,dt
            \le \int_0^{9q\ell/{\ln\ell}} 1 \, dt = \frac{9 q \ell}{\ln \ell}.
            \label{eq:integral-1}
        \end{equation}

        \item
        In the next region,
        $t$ is now large enough that
        our tail bound
        becomes meaningful.
        However, this integral
        does not converge,
        so we will need an appropriate upper limit.
        This second region
        bridges the gap
        until very strong tail bounds hold.
        \Cref{thm:remove-surject} applies in this regime
        if $9 \ell/{\ln \ell} \ge 8$,
        which is true for $\ell \ge \Omega(1)$.
        \begin{align}
        \int_{9q\ell/{\ln\ell}}^{\gamma}
        \Pr[M(S,h)\ge t]\,dt
            &\le
        \int_{9q\ell/{\ln\ell}}^{\gamma}
        \frac{8}{t} \exp\parens*{\frac{8 q (\ell + 2\ceil{\log_q
              t})}{t}} \, dt.
              \end{align}
      The exponential term
          is nonincreasing for $t \ge e$,
          and our integral's limits clearly fall in this regime.
          So, let us simplify the integral
          by moving the exponential term out
          with the lower limit substituted.
          This term on its own is upper bounded by
          \begin{align}
              \exp\parens*{\frac{8(\ell + 2\log_q(8 \ell/{\ln
              \ell}) + 4)}{9\ell/{\ln \ell}}}
              \le           \exp\parens*{\frac{8 \ln \ell}{9}
              + \frac{16 \log_2(8 \ell/{\ln \ell})}{9\ell/{\ln \ell}}
              + \frac{32}{9\ell/{\ln \ell}}}
              \le
              \ell^{9/10},
          \end{align}
          for $\ell$
          larger than an absolute constant.
          Therefore,
          \begin{align}
          \int_{9q\ell/{\ln\ell}}^{\gamma}
        \Pr[M(S,h)\ge t]\,dt
            &\le 8 \ell^{9/10}
              \int_{9q\ell/{\ln \ell}}^\gamma \frac 1 t \, dt\\
        &\le 8 \ell^{9/10} \ln \gamma.\label{eq:integral-2}
        \end{align}
        Choose
        $\gamma \coloneq \exp(q \ell^{1/20})$
        so that this
        is upper bounded by
        $8 q \ell^{19/20}$.

        \item Lastly,
        once $t$ is large enough,
        we can utilize the extremely strong tail bound
        \cref{thm:ADMPT-bound}
        from \cite{admpt}.
        In this region,
        the integral
        will easily converge
        to something vanishing.
 \begin{align}
            &\int_{\gamma}^\infty
            \Pr[M(S,h)\ge t]\,dt\\
            &\le
            \int_{W}^\infty
            \Pr[M(S,h)\ge wq^{15 + \log_2\log_2 q }\ell^{\log_2 q}] 
            (q^{15}\ell^{\log_2 q})
            \, dw\\
            &\le
            2q^{15 + \log_2 \log_2 q}\ell^{\log_2 q}
            \int_{W}^\infty           \frac{1}{w^{\log_q w}} \, dw\\
            &\le
            2q^{15 + \log_2 \log_2 q}\ell^{\log_2 q}
            \int_{W}^\infty
            \frac{1}{w^{\log_q W}} \, dw\\
            &\le \frac{
            \exp(O( \ln q(\ln \ell + \log_2\log_2 q)))
            }{(\log_q W - 1) W^{\log_q W - 1}}.\\
            \intertext{
        Here, $W \coloneq \gamma/(q^{15+\log_2\log_2 q} \ell^{\log_2 q}) = \exp(q\ell^{1/20} - O(\ln q(\ln \ell + \log_2\log_2 q))) \ge q^{\ell^{1/20} - O(\ln \ell )}$, which gives}
        &\le \frac
        {q^{O(\ln \ell)}}
        {
        q^{(\ell^{1/20} - O(\ln \ell))^2}},
        \end{align}
        which is at most $1$
        for $\ell$ larger than an absolute constant.
    \end{enumerate}
    In total, whenever $\ell$ is larger than some absolute constant $\ell_0$,
    the expectation is upper bounded by
    \[
\E_h[M(S,h)] \le
\frac{9 q \ell}{\ln\ell}
+ 8q\ell^{19/20}
+ 1
\le {\frac{10q \ell}{\ln \ell}}.
    \]
    We can also write
    this bound in terms of $n$
    instead of $\ell$ as follows:
    \begin{align}
        \frac{10q \ell}{\ln \ell}
        &= \frac{10q}{\ln q}
        \cdot \frac{\ln n}{\ln \ln n - \ln \ln q}\\
        &= \frac{10q}{\ln q}
        \cdot \frac{\ln n}{\ln \ln n}
        \cdot \frac{\ln (\ell \ln q)}{\ln (\ell \ln q) - \ln \ln q}\\
        &= \frac{10q}{\ln q}
        \cdot \frac{\ln \ell + \ln \ln q}{\ln \ell} \cdot \frac{\ln n}{\ln \ln n}\\
        &= O\parens*{\frac{q \ln \ln q}{\ln q} \cdot {\frac{\ln n}{\ln \ln n}}}.
    \end{align}
    (When $q=2$,
    $\ln \ln q$
    is actually negative, but we still recover a $O(\ln n/{\ln \ln n})$
    upper bound.)
\end{proof}

\section{List Decodability of Random $\F_q$-Linear Codes}

An error-correcting code $C \subseteq \F_q^n$
(with relative distance $\delta$, say)
has the property that every two codewords 
differ in at least $\delta n$ places.
Equivalently, every Hamming ball of radius $\delta n$ contains at most one codeword of $C$.
List decoding generalizes this property to larger Hamming balls.

\begin{definition}
    For $0 \le p < 1-1/q$,
    a \emph{$(p, L)$-list decodable code}
    $C \subseteq \F_q^n$
    satisfies the property that
    every Hamming ball of radius $pn$
    contains at most $L$ codewords of $C$.
\end{definition}

Let $H_q(p)$ denote the $q$-ary entropy function.
The \emph{list decoding capacity} $1 - H_q(p)$
is a special rate threshold
for when
list decodability with
constant list size is achievable.
We recommend Chapter~7 of 
Guruswami, Rudra, and Sudan \cite{GRS}
for more details.

\begin{theorem}[List decoding capacity,
see {\cite[Section 7.4]{GRS}}]
   Let $0 \le p < 1- 1/q$
   and let $\eps > 0$ be sufficiently small.
   Then for every sufficiently large $n$,
   \begin{enumerate}[(1)]
   \item
   there exists a
   $(p, O(1/\eps))$-list decodable code
   with rate
   $1 - H_q(p) - \eps$, and
   \item 
   every $(p, L)$-list decodable code with rate $1 - H_q(p) + \eps$
   must have $L \ge q^{\Omega(\eps n)}$.
   \end{enumerate}
\end{theorem}

We now show that random $\F_q$-linear codes achieve list decoding capacity---in particular, have optimal list size near this rate threshold---for $q = O(1)$, by directly showing list size is a special case 
of hashing max load.
We will define a random linear code (RLC) $C \subseteq \F_q^n$ of rate $R$
as the kernel of a uniformly random
parity check matrix 
of size $(1-R)n \times n$.
Up to an
exponentially small probability, this is equivalent to choosing the code
as a uniformly random
generator matrix $n \times Rn$ or as a uniformly random subspace of dimension $Rn$. 
First, show that bounding the list size of a linear code is equivalent to bounding the max load when linearly hashing a Hamming ball. This connection was briefly mentioned in \cite{dlmnrr}, and we formally prove it here.

\begin{lemma}
\label{lem:ls-hashing}
    Let $h:\F_q^n\to\F_q^{(1-R)n}$ be a linear map, $C\coloneqq h^{-1}(0)$, and $B\subset \F_q^n$ be the Hamming ball of radius $pn$ centered at $0$. $C$ is $(p,L)$-list decodable if and only if $M(B,h)\le L$. 
\end{lemma}

\begin{proof}
     The definition of $(p,L)$-list decodability of $C$
     can be written as 
     \[
     \max_{x \in \F_q^n} {\abs{(x + B) \cap C}}
     \le L.
     \]
     However, by relativity of motion,
     this is
     equivalent to 
    \begin{equation}
    \label{eq:list-size}
     \max_{x \in \F_q^n} {\abs{B \cap (x + C)}}
     \le L.
     \end{equation}
     Finally, 
     the left hand side
     is equal
     to the max load $M(B,h)$.
     We can tile $\F_q^n$
     with shifts of $C$,
     and each shift 
     corresponds to a different
     nonempty preimage of $h$.
     We can extend the
     quantifier to consider
     all preimages
     without changing
     the maximum value.
     \begin{align}
        \max_{x \in \F_q^n} {\abs{B \cap (x + C)}}
        &= \max_{y \in \F_q^n}
        {\abs{B \cap h^{-1}(y)}} \le L.\qedhere
     \end{align}

\end{proof}
Now that we know list size
is equivalent to max load,
we can plug the set $B$
into our tail bound.

\begin{theorem}[formal \cref{thm:intro-list-decoding}]
\label{thm:list-decoding-final}
   Consider all $0 \le p < 1- 1/q$
   and all $0 < \eps < 1-H_q(p)$.
    Let $C\le \F_q^n$
    be a random linear code 
    of rate $1-H_q(p)-\eps$,
    and let 
    \[
    L \coloneq 
    {16q}%
    \cdot \parens*{\frac{H_q(p)}{\eps}
    + 1}.
    \]
    If $n \ge 40/\eps \cdot \log_q (1/\eps)$,
    then $C$ is $(p, L)$-list decodable with probability $1-q^{-\Omega(\eps n)}$.
\end{theorem}

\begin{proof}
   Let $B\subset \F_q^n$ be the Hamming ball of radius $p n$ centered about zero, so that $\abs{B}\le q^{H_q(p) n}$, and let $R\coloneq 1-H_q(p)-\eps$.
    A random linear code $C$ of rate $R$
    is defined as the kernel of a 
    uniformly random linear function 
    $h\colon \F_q^n \to \F_q^{(1-R)n}$.

    We can bound the probability that $C$ is \emph{not} $(p, L)$-list decodable.
    By \cref{lem:ls-hashing}, 
    this is equivalent to the max load 
    $M(B,h)$ being larger than $L$,
    and we can show this is rare using
    \cref{thm:remove-surject}.
    \begin{align*}
    \Pr[\text{$C$ is not $(p,L)$-list decodable}] & =  \Pr\left[M(B,h) > L\right] \\ 
    &\leq
  \frac{9\abs{B}}{q^{(1-R)n} L} \cdot \exp\parens*{\frac{8q((1-R)n + 2\ceil{\log_q L})}{L}}\\
    & \le \frac{9q^{H_q(p) n}}{q^{(H_q(p) + \eps)n }L}\cdot \exp\parens*{\frac{8q((H_q(p) + \eps)n + 2\ceil{\log_q L})}{L}} \\
    &  = \frac{9}{q^{\eps n}L}\cdot \exp\parens*{\frac{8q((H_q(p) + \eps)n + 2\ceil{\log_q L})}{L}}.\end{align*} 
    To control the exponential term,
    let us set list size to be
    \[L = 
    \frac{16q(H_q(p) + \eps)}
    {\eps}%
    = %
    {16q}%
    \cdot \parens*{\frac{H_q(p)}{\eps}
    + %
    {1}
    },\]
    which is at least $8q$
    as needed to apply
    \cref{thm:remove-surject}. %
    Then,
    \begin{align*}
    &\Pr\braks*{\text{$C$ is not $(p,L)$-list decodable}}\\
    &\le \frac{9 \eps}{16 q H_q(p) + 16q \eps}
    \exp\parens*{\frac{16 \eps \ceil{\log_q L}}{16H_q(p) + 16 \eps}}\cdot e^{\eps n / 2} \cdot q^{-\eps n}\\
    &\le \frac{9\eps}{16q\eps}
    \exp\parens{\ceil{\log_q L}}\cdot q^{-\eps n/4}\\
    &\le 
    \frac{144 e^2}{q \eps^2}
    \cdot q^{-\eps n / 4}.
    \intertext{If
    $n \ge 40/\eps \cdot \log_q (1/\eps)$, then
    $\eps^2 q^{\eps n / 4} \ge q^{\eps n / 5}$, so
    this is
    upper bounded by}
    &\le  \frac{144e^2}{q^{\eps n / 5}}
        = q^{-\Omega(\eps n)}.\qedhere
    \end{align*}
    
\end{proof}

\section*{Acknowledgements}
VMK thanks Ray Li for various early conversations regarding list decodability and Kakeya sets.
We thank Sidhanth Mohanty for many conversations during Summer 2026
and for hosting us at Northwestern.
We thank Venkatesan Guruswami for discussions and pointers to literature.

\paragraph{AI involvement.}
All ideas, analysis, and writing
here
are due to the authors.
No AI was used.

\bibliographystyle{alphaurl}
\bibliography{references}

\appendix

\section{Upper Bound via Probabilistic Method}
\label{apx:tight}
When $\delta=1$,
we can show that
\cref{cor:furst-constant-density}
is almost tight
(up to a $q$ factor
in the exponential term)
using random sets.
We begin with a generic upper bound,
and then set parameters appropriately
in both lower and upper bounds. Any density-$\alpha$ set is a $(1,n,r,\alpha q^r)$-Furstenberg set by pigeonhole over the shifts of each subspace,
    so we will hope to beat this using the probabilistic method.
    We do not attempt whatsoever to
    optimize the constants here.
\begin{lemma}
\label{lem:divergence}
For sufficiently large $n$,
let $r < n/2$ and
let $S\subseteq \F_q^n$
be drawn i.i.d.\ randomly
with probability $p < \alpha$.
For every $\eta < 1$,
suppose that
\begin{equation}
    \label{eq:divergence-bound}
    q^r D(\alpha \mathbin{\|} p)
    <
    {(n - 3r/2) \ln q - \ln{\sqrt{2}}
    - \ln \ln(q^{r(n-r)}/\eta)}.
\end{equation}
Then
$S$ is a $(1,n,r,\alpha q^r)$-Furstenberg set with probability at least $1-\eta$.
\end{lemma}
\begin{proof}
   The probability that $S$
   is \emph{not} Furstenberg
   can be estimated
    by union bounding
    over each dimension-$r$ subspace $V$
    and considering the probability
    of every shift of $V$
    being small in $S$.
    \begin{align}
    \label{eq:random-furstenberg}
        \Pr[S\text{ is not Furstenberg}]
        &\le q^{r(n-r)}(1-\Pr[\abs{V \cap S} \ge \alpha q^r])^{q^{n-r}}.
    \end{align}
    $\abs{V \cap S}$
    is binomially distributed with expectation $p q^r < \alpha q^r$,
    and its tail satisfies~\cite[(4.7.2)]{Ash}
    \begin{align}
    \Pr[\abs{V \cap S} \ge \alpha q^r]
    &\ge \frac{\exp\parens*{
    - q^r D( \alpha \mathbin{\|} p )
    }
    }{\sqrt{8 q^r \alpha (1-\alpha) }}\\
    &\ge
    \frac{1
    }{\sqrt{2 q^r}}
    \cdot
    \exp\parens*{
    - q^r D( \alpha \mathbin{\|} p )
    }.
    \end{align}
    Substitute into \eqref{eq:random-furstenberg}
    and upper bound using
    the exponential inequality.
    \begin{align}
        \Pr[\text{$S$ is not Furstenberg}]
        &\le
        q^{r(n-r)}
        \exp\parens*{
        -\frac{q^{n-r}}
        {\sqrt{2 q^r}}
        \exp(-q^r D(\alpha \mathbin{\|} p))
        }.
    \end{align}
    If $p$
        satisfies
        \eqref{eq:divergence-bound},
        then
    this entire expression is strictly less than $\eta$.
\end{proof}

\begin{proposition}[Tightness]
    Let $\alpha\in (0,1)$, $q,r\in \Z$,
     $r < n/2$,
     and $n$ sufficiently large.
    \begin{enumerate}[(1)]
        \item If $n = \Omega(\alpha q^{r-1})$, the smallest $(1,n, r, \alpha q^r)$-Furstenberg set $S$ satisfies
        \[
        \frac{\alpha}{2} \exp\parens*{-\frac{4n}{\alpha q^{r-1}}}
        \le
        \frac{\abs{S}}{q^n} \le 3\alpha\exp\parens*{-\frac{n-O(r)}{\alpha q^r} \cdot \ln q}.\]

        \item If $n= o(\alpha q^{r-1})$, the smallest $(1,n,r,\alpha q^r)$-Furstenberg set $S$ satisfies
        \[
        \alpha - 3\sqrt{\frac{\alpha n}{q^{r-1}}} \le
        \frac{\abs{S}}{q^n} \le\alpha - \sqrt{\frac{2 \alpha(n - O(r)) \ln q}{q^r}} + o(1).\]
    \end{enumerate}
\end{proposition}
\begin{proof}
    Let us begin by calculating
  a Furstenberg upper bound
    using the probabilistic method,
    and then set parameters appropriately
    for lower and upper bounds
    in both regimes.

    \paragraph{Large $n$.}
    When $n = \Omega(\alpha q^{r-1})$,
    the exponential term of
    \cref{cor:furst-constant-density}
    dominates.
    Set $C = 2$
    to get
    \[
    \frac{\abs{S}}{q^n}
    \ge \frac{\alpha}{2} \cdot \exp\parens*{-\frac{4n}{\alpha q^{r-1}}}.
    \]
    For the upper bound,
    choose $p$
    to be
    \begin{align}
    p
    &= \alpha (1-\alpha)^{1/\alpha}
    \exp\parens*{
    - \frac{(n - 3r/2) \ln q - \ln{\sqrt{2}}
    - \ln \ln (3q^{r(n-r)})}
    {\alpha q^r}
    }\\
    &\le \frac{\alpha}{2} \cdot \exp\parens*{
    -\frac{n - O(r)}{\alpha q^{r}}
    \cdot \ln q
    }.
    \end{align}
    This choice of $p$
    satisfies \cref{lem:divergence}
    for $\eta = 1/3$,
    using the fact that
    \begin{align}
    D(\alpha \mathbin{\|} p)
    &= {\alpha} \ln{\frac{\alpha}{p}}
    + (1-\alpha) \ln{\frac{1-\alpha}{1-p}}
    \ge
    \alpha \ln{\frac{\alpha}{p}}
    - \ln\frac{1}{1-\alpha}.
    \end{align}
    So, $S$ chosen i.i.d.\ with
    probability $p$ will be Furstenberg with probability at least $2/3$.
    Next,
    by Markov's inequality,
    $\abs{S} \le 3p$
    with probability at least $2/3$.
    Hence,
    $\abs{S} \le 3p$
    and $S$ is Furstenberg
    with probability at least $1/3$,
    so there exists
    a
    $(1,n,r,\alpha q^r)$-Furstenberg
    set with density at most $3p$.

    \paragraph{Small $n$.}
    When $n = o(\alpha q^{r-1})$,
apply the exponential inequality
to
\cref{cor:furst-constant-density}
and balance the first two terms by setting $C = \sqrt{\alpha q^{r-1} / n}$
to get
\begin{align}
    \frac{\abs{S}}{q^n}
    &\ge
    \alpha - 3\sqrt{\frac{\alpha n}{\delta q^{r-1}}}.
\end{align}
For the upper bound,
use the approximation
\[
D(\alpha \mathbin{\|} p)
\ge \frac{(\alpha - p)^2}{2\alpha}.
\]
Then,
we can satisfy \cref{lem:divergence}
as long as $p$ satisfies
\begin{align}
    \alpha - p
    &\le
    \sqrt{\frac{(2\alpha)((n-3r/2)\ln q
    - \ln{\sqrt{2}}
    -\ln \ln(3q^{r(n-r)}))}{q^r}}.
\end{align}
Hence, it suffices to choose
\[
p = \alpha - \sqrt{\frac{2\alpha (n - O(r)) \ln q}{q^r}} + o(1),
\]
and $S$ chosen i.i.d.\ with probability $p$
will be Furstenberg with probability
at least $2/3$.
In addition,
via the Chernoff bound,
the density of $S$
will be at most $p + O(\sqrt{p/q^n})$
with probability at least $2/3$.
So, $S$ will be Furstenberg and also
have size
\[\abs{S} \le \alpha - \sqrt{\frac{2\alpha (n - O(r)) \ln q}{q^r}} + O\parens*{\sqrt{\frac{\alpha}{q^n}}} + o(1).\qedhere\]
\end{proof}

\section{Tail Bound of Alon--Dietzfelbinger--Miltersen--Petrank--Tardos}
\label{sec:admpt}
\begin{lemma}
\label{lem:square-decay-tail}
Let $\alpha_i$ for $1\le i\le k$ be random variables and let $0 < \alpha_0 < 1$ be a constant.. Suppose that for $0\le i < k$ we have $0\le \alpha_{i+1}\le \alpha_i$ and conditioned on any set of values for $\alpha_1,\dots, \alpha_i$, we have $\E[\alpha_{i+1} \mid \alpha_1,\dots, \alpha_i]\le \alpha_i^2$. Then
for any threshold $0 < t < 1$ we have
\[\Pr[\alpha_k\ge t] \le \alpha_0^{k - \log_2 \log_2(1/t) + \log_2\log_2(1/\alpha_0)}.\]
\end{lemma}

From this, a weak version of coupon collection can be shown.

\begin{lemma}
\label{lem:covering}

    \begin{enumerate}[(a)]
    \item Let $S\subset \F_q^u$ of density $\frac{|S|}{q^u}\coloneq 1-\alpha < 1$ and $0\le \ell \le u$ is an integer, then for uniform random surjective linear map $h: \F_q^u\to\F_q^\ell$, \[\Pr[h(S)\neq \F_q^\ell]\le  \alpha^{u - \ell - \log_2 \ell - \log_2 \log_2 q + \log_2\log_2(1/\alpha)}.\]

    \item There exists constant $c > 0$ such that the following holds. Let $S$ be a subset of a vector space $V$ of arbitrary dimension over $\F_q$. Let $\ell > 0$ be an integer. If $|S|\ge q^{15}\ell^{\log_2 q}q^\ell $, then for a uniform random linear $h: V\to \F_q^\ell$, \[\Pr[h(S) = \F_q^\ell] \ge \frac{1}2.\]
    \end{enumerate}
\end{lemma}

We note that when $q=2$, item (b) recovers the coupon collector prediction, but for larger $q$, the set $S$ will be asymptotically larger (even for $q=3$). This is the cause behind the max load degradation for larger $q$.

\begin{proof}
\begin{enumerate}[(a)]
    \item Denote $k\coloneqq u-\ell$. We pick $h$ by picking $k$ vectors $v_1,\dots, v_k$ uniformly at random from $\F_q^u$, and choosing $h$ to be a random surjective linear map conditioned on the kernel of $h$ containing $v_1,\dots, v_k$. Denote $A_0 = A$,  $A_{i+1} = A_i + \text{span}\{v_i\} $, $\conj{A_i} = \F_q^u\setminus A_i$, and $\alpha_i\coloneq |\conj{A_i}|/q^u$. Observe that \[\conj{A_{i+1}} = \bigcap_{\gamma\in\text{span}\{v_i\}}(\conj{A_i} + \gamma)\subset \conj{A_i}\cap (\conj{A_i}+v_{i+1}).\] Hence, conditioned on a choice of $v_1,\dots, v_i$, we can bound \[\E[\alpha_{i+1}] = \Pr_{x\sim \F_q^u} [x\in \conj{A_{i+1}}] \le  \Pr[x\in  \conj{A_i}\cap (\conj{A_i} + v_{i+1})]= \frac{1}{q^u}\sum_{x\in \conj{A_i}} \Pr[x-v\in \conj{A_i}] = \alpha_i^2.\]
    Therefore, the $\alpha_i$ satisfy the assumption of \cref{lem:square-decay-tail}, and thus \[\Pr[\alpha_k \ge q^{-\ell}]\le \alpha^{k-\log_2\log_2(q^\ell)+ \log_2\log_2(1/\alpha)} = \alpha^{u - \ell - \log_2 \ell - \log_2 \log_2 q + \log_2\log_2(1/\alpha)}.\] It suffices to show that $\alpha_k < q^{-\ell}$ implies $h(S) = \F_q^\ell$. Now define $B = h(A)$. Since $h$ is surjective, we have $|h^{-1}(B)| = q^{u-\ell} |B|.$ Now note that $h(A_k) = h(A + \text{Span}(v_1,\dots, v_k)) =  h(A + \ker h) = B$. Therefore $A_k\subset h^{-1}(B)$, and so $|A_k|\le q^{u-\ell}|B|$. However, the bound on $\alpha_k$ implies $|A_k| > q^u - q^{u-\ell}$. Hence, we must have $|B| > q^{\ell} - 1$, and thus $|B| = q^\ell$. This implies $h(A) = \F_q^\ell$ as desired.
    \item Let $W = \F_q^{t}$, where $t = \lfloor \log_q(4m)\rfloor$. We factor $h$ into $h_1\circ h_0$, where $h_0: V\to W$ is a random linear map, and $h_1: W\to \F_q^\ell$ is a random surjective linear map. For any fixed distinct $s,t\in S$, $h_0$ will collide the two with probability $1/|W|$. Therefore, the expected number of collisions caused by $h_0$ is at most \[\binom{m}2\cdot \frac{1}{|W|}\le \frac{m^2}2 \cdot \frac{1}{4m} = \frac{m}8.\] Hence, by Markov's inequality, the number of collisions exceeds $m/2$ with probability at most $1/4$. With that many collisions, it will be that $|h_0(S)|\le m/2.$ Condition on the event that this doesn't happen, and $S_0\coloneqq h_0(S)$ is of size $\ge m/2$. Denote $\alpha = 1-|S_0|/|W|$. We observe \[\alpha \le 1-\frac{m/2}{4m} = \frac{7}8.\] Hence, by part (a), it follows \begin{align*}\Pr[h_1(S_0) \neq  \F_q^\ell] & \le \left(\frac{7}8\right)^{t - \ell - \log_2 \ell - \log_2\log_2 q + \log_2\log_2(8/7)} \\ & \le \exp\left(-\frac{1}8(\log(4m) - \ell - \log_2 \ell - \log_2\log_2 q + \log_2\log_2(8/7))\right).\end{align*} Let $C_1 \coloneqq \log_2\log_2(8/7)$ and $C_2\coloneqq 8\ln 4$. Hence, if we pick $m$ such that \[\log_q(4m) - \ell - \log_2 \ell - \log_2\log_2 q + C_1 = C_2,\] it follows that $\Pr[h_1(S_0)\neq \F_q^\ell]\le \exp(-\frac{1}8 \cdot 8\ln 4) < 1/2$. This means\begin{align*}\log_q(4m) & = C_2-C_1+\log_2 \ell + \ell + \log_2\log_2 q \\ \implies m & = \frac{1}4 q^{C_2-C_1+\log_2 \ell + \ell + \log_2\log_2 q}\le q^{15 + \log_2\log_2 q}q^{\log_2 \ell + \ell}.\end{align*} For such an $m$, it follows \[\Pr[h(S)\neq \F_q^\ell] = \Pr[|S_0| < m/2] + \Pr\left[h_1(S_0)\neq \F_q^\ell \;\bigg|\; |S_0| \ge m/2\right]\le \frac{1}4 + \frac{1}4 = \frac{1}2.\qedhere\]
\end{enumerate}
\end{proof}

With these covering properties, a max load tail bound can be derived.

\begin{theorem} Let $t\ge 1$, $u\ge \ell\ge 1$ be integers, and let $S\subset \F_q^u$ be an arbitrary subset of size $q^\ell$. For uniformly random linear $h:\F_q^u\to\F_q^\ell$, 
\label{thm:admpt-t-tail}
    \[
    \Pr[ M(S,h) \ge q^{15+\log_2\log_2 q + t + \log_2 t}] < 2q^{-t(t - \log_2 \ell + \log_2 t)}.
    \]
\end{theorem}

\begin{proof}
Denote $T\coloneqq q^{15+\log_2\log_2 q + \log_2 t}q^t$. We will factor through $\F_q^{\ell + t}$. Let $h_1:\F_q^u\to\F_q^{\ell + t}$ be a random linear map, and let $h_2: \F_q^{\ell + t}\to \F_q^\ell$ be a random surjective linear map.
Denote the event \[\calE = \{\exists y\in \F_q^\ell: h_2^{-1}(y)\subset h_1(S)\}.\] We first show that $\Pr[\calE \mid M(S,h)\ge T]\ge 1/2$. Condition on $h$ and $h_2$ such that  $|h^{-1}(y)\cap S|\ge T$ for some $y$. Consider the distribution over all $h_1$ after this conditioning. We will show that the probability $\calE$ is satisfied is at least $1/4$, which implies the same lower bound in the weaker conditioning of $M(S,h)\ge T$.  Denote $V \coloneqq h^{-1}(y)$ and $W \coloneqq h^{-1}_2(y)\cong \F_q^t$.  Let $h_1' = h_1|_{V}$, which we note is a map taking $V$ to $W$. By Proposition~3.4 of \cite{admpt}, $h_1'$ is either a random linear or random affine linear map. Let $S' = S\cap V$. Applying \cref{lem:covering} to $h_1'$ and $S'\subset V$, it follows in this conditional distribution, \[\Pr[\calE] \ge \Pr\braks{h_2^{-1}(y)\subset h_1(S)} \ge \Pr\braks{h_1'(S') = \F_q^t}\ge \frac{1}2.\]

    Therefore, $\Pr[\calE \mid M(S,h)\ge T]\ge 1/2$, and so we have \begin{equation}\label{eq:admpt-tail}\Pr[M(S,h)\ge T] = \frac{\Pr[\calE]}{\Pr[\calE \mid M(S,h)\ge T]}\le 2 \Pr[\calE].\end{equation} To estimate $\Pr[\calE]$, we phrase it as a covering problem. Let $S_1\coloneqq \F_q^{\ell + t}\setminus h_1(S)$. We see \[\calE\iff h_2(S_1) \neq \F_q^\ell.\]  Therefore, by \cref{lem:covering} and the fact that $1-|S_1|/q^{\ell + t} \le q^{\ell}/q^{\ell + t} = q^{-t} $, we have \[\Pr[\calE]\le q^{-t(t - \log_2 \ell - \log_2\log_2 q+ \log_2\log_2(q^{t}))} = q^{-t(t - \log_2 \ell + \log_2 t)}.\] Combining this with \eqref{eq:admpt-tail} yields the desired result.
\end{proof}

\begin{corollary}
\label{thm:ADMPT-bound}
 Let $u\ge \ell\ge 1$ be integers, and let $h: \F_q^u\to\F_q^\ell$ be a random linear map. For arbitrary $S\subset \F_q^u$ of size $q^\ell$, we have
    \[\Pr[M(S,h) > wq^{15 + \log_2 \log_2 q}\ell^{\log_2 q}] \le 2w^{-\log_q w}.\]
\end{corollary}

\begin{proof}
    Set $t$ to be the smallest integer such that $t + \log_2 t \ge \log_2\ell + \log_q w + \log_2\log_q w$. Thus, $t \ge \log_q w$, and plugging this $t$ into \cref{thm:admpt-t-tail},
    \[\Pr[M(S,h) \ge wq^{15 + \log_2 \log_2 q}\ell^{\log_2 q}]\le 2q^{-t \log_q w} = 2w^{-t}\le 2w^{-\log_q w}.\qedhere\]
\end{proof}
It is not hard to see that integrating this tail bound yields the following expectation bound.
\begin{corollary}

    \[\E_h[M(S,h)] = O(q^{15+\log_2\log_2 q}\ell^{\log_2 q}).\]
\end{corollary}